\documentclass{llncs}
\usepackage[english]{babel}

\usepackage{microtype}%if unwanted, comment out or use option "draft"

\usepackage[autostyle=true]{csquotes} % Required to generate language-dependent quotes in the bibliography

\usepackage{amsmath}
\usepackage{amssymb}

\usepackage{enumitem}

\usepackage{algpseudocode}
\usepackage{algorithm}

\usepackage{color}
\usepackage[toc,page]{appendix} % for the appendix
\usepackage{xstring}

\usepackage{tikz}
\usetikzlibrary{positioning}
\usetikzlibrary{backgrounds}
\usetikzlibrary{arrows,shapes,snakes,automata,backgrounds,petri,calc,patterns}
\usetikzlibrary{petri,arrows,backgrounds,fit,positioning,automata,matrix,shapes}
\usepackage[absolute,overlay]{textpos}
\usepackage{subfig}
\usepackage{enumitem}

\tikzset{>=latex}

\newcommand{\FFF}{\mathbf{F}}
\newcommand{\GGG}{\mathbf{G}}

\newcommand{\GG}[1]{\GGG \ ({#1})}
\newcommand{\FF}[1]{\FFF \ ({#1})}

\newcommand{\PG}[2]{\GGG^{#2}{#1}}
\newcommand{\PF}[2]{\FFF^{#2}{#1}}

\newcommand{\FGfrag}[2]{$\mathbf{F}_{#1} \mathbf{G}_{#2}$}

\newcommand{\atomic}{\mathcal{A}}
\newcommand{\literals}{\mathcal{L}}

\newenvironment{proofsketch}{\proof}{\endproof}

\newcommand{\F}[1]{\ensuremath{\mathbf{F}_{#1}}}
\newcommand{\G}[1]{\ensuremath{\mathbf{G}_{#1}}}
\newcommand{\U}[1]{\ensuremath{\mathbf{U}_{#1}}}
\newcommand{\X}[1]{\ensuremath{\mathbf{X}_{#1}}}

\newcommand{\initstate}{s_0}
\newcommand{\pr}{\mathbb P}
\newcommand{\post}{\mathit{post}}
\newcommand{\runs}{\mathit{Runs}}

\allowdisplaybreaks
\newcommand{\QED}{\hfill\ensuremath{\square}} % quod erat exemplandum

\newcommand{\thmhelperpre}[2]{\newcommand{\theoremlike}[1]{\par\medskip\penalty-250\refstepcounter{theorem}{\bfseries\noindent##1 \ref{#1}.}\itshape}\theoremlike{#2}}
\newcommand{\thmhelperpost}{\par\medskip%
	\renewcommand{\theoremlike}[1]{\par\medskip\penalty-250\refstepcounter{theorem}{\bfseries\noindent##1 \thesection .\thetheorem.}\itshape}%
}

\newenvironment{refcorollary}[1]{\thmhelperpre{#1}{Corollary}}{\thmhelperpost}
\newenvironment{reftheorem}[1]{\thmhelperpre{#1}{Theorem}}{\thmhelperpost}
\newenvironment{reflemma}[1]{\thmhelperpre{#1}{Lemma}}{\thmhelperpost}

\newcommand{\todowidth}{2.5cm}
\newcommand{\todo}[1]{
	\begin{tikzpicture}[remember picture, baseline=-0.75ex]
	\node [coordinate] (inText) {};
	\end{tikzpicture}
	\marginpar{
		\begin{tikzpicture}[remember picture, font=\scriptsize]
		\draw node[draw=black, text width = \todowidth, inner sep=0.1mm] (inNote){#1};
		\end{tikzpicture}
	}
	\begin{tikzpicture}[remember picture, overlay]
	\draw[draw=red]
	([yshift=-0.2cm] inText)	-| (inNote.south);
	\end{tikzpicture}}
\renewcommand{\todo}[1]{}

\begin{document}
\title{The Satisfiability Problem for Unbounded Fragments of Probabilistic CTL
\thanks{This research was funded in part by the Czech Science Foundation grant No.~\mbox{P202/12/G061}, TUM IGSSE Grant 10.06 (PARSEC), and the German Research Foundation (DFG) project 383882557 ``Statistical Unbounded Verification''. We would like to thank Tobias Meggendorfer and the anonymous reviewers for their helpful comments on the draft.}}

\titlerunning{The Satisfiability Problem for Unbounded Fragments of PCTL}%optional, please use if title is longer than one line

\author{Jan K\v ret\'insk\'y \inst{1} \and Alexej Rotar \inst{2}}

\institute{Technical University of Munich, Germany \email{jan.kretinsky@tum.de} \and Technical University of Munich, Germany \email{alexej.rotar@tum.de}}

\authorrunning{J. K\v ret\'insk\'y and A. Rotar}

%\Copyright{Jan K\v ret\'insk\'y and Alexej Rotar}%mandatory, please use full first names. LIPIcs license is "CC-BY";  http://creativecommons.org/licenses/by/3.0/

%\subjclass{Theory of computation $\rightarrow$ Logic $\rightarrow$ Modal and temporal logics}% mandatory: Please choose ACM 2012 classifications from https://www.acm.org/publications/class-2012 or https://dl.acm.org/ccs/ccs_flat.cfm . E.g., cite as "General and reference $\rightarrow$ General literature" or \ccsdesc[100]{General and reference~General literature}. 

%\category{}%optional, e.g. invited paper

%\relatedversion{}%optional, e.g. full version hosted on arXiv, HAL, or other respository/website

\maketitle

\pagestyle{plain}

%\frontmatter

\begin{abstract}
	We investigate the satisfiability and finite satisfiability problem for probabilistic computation-tree logic (PCTL) where operators are not restricted by any step bounds. We establish decidability for several fragments containing quantitative operators and pinpoint the difficulties arising in more complex fragments where the decidability remains open.
%\keywords{temporal logic \and probabilistic verification \and probabilistic computation tree logic \and satisfiability}
\end{abstract}

%\mainmatter

\section{Introduction}
\label{Chapter1}

Temporal logics are a convenient and useful formalism to describe behaviour of dynamical systems.  Probabilistic CTL (PCTL) \cite{HS,HJ:logic-time-probability-FAC} is the probabilistic extension of the branching-time logic CTL \cite{EH82}, obtained by replacing the existential and universal path quantifiers with the probabilistic operators, which allow us to quantify the probability of runs satisfying a given path formula.  At first, the probabilities used were only 0 and 1 \cite{HS}, giving rise to the \emph{qualitative PCTL (qPCTL)}. This has been extended to any values from [0, 1] in \cite{HJ:logic-time-probability-FAC}, yielding the \emph{(quantitative) PCTL} (onwards denoted just \emph{PCTL}). More precisely, the syntax of these logics is built upon atomic propositions, Boolean connectives, temporal operators such as \textbf{X} (``next'') and \textbf{U} (``until''), and the probabilistic quantifier ${\bowtie q}$ where $\bowtie$ is a numerical comparison such as $\leq$ or $>$, and $q\in[0, 1]\cap \mathbb Q$ is a rational constant.  A simple example of a PCTL formula is $\mathit{ok} \U{=1}(\X{\geq0.9} \mathit{finish})$, which says that on almost all runs we reach a state where there is 90\% chance to \textit{finish} in the next step and up to this state \textit{ok} holds true. PCTL formulae are interpreted over Markov chains \cite{Norris:book} where each state is assigned a subset of atomic propositions that are valid in a given state.

In this paper, we study the \emph{satisfiability problem}, asking whether a given formula has a \emph{model}, i.e. whether there is a Markov chain satisfying it. If a model does exist, we also want to construct it. Apart from being a fundamental problem, it is a possible tool for checking consistency of specifications or for reactive synthesis. The problem has been shown EXPTIME-complete for qPCTL in the setting where we quantify over finite models (\emph{finite satisfiability}) \cite{HS,LICS} as well as over generally countable models (\emph{infinite satisfiability}) \cite{LICS}.  The problem for (the general quantitative) PCTL remains open for decades.  We address this question on fragments of PCTL.  In order to get a better understanding of this ultimate problem, we answer the problem for several fragments of PCTL that are
\begin{itemize}
	\item quantitative, i.e.\ involving also probabilistic quantification over arbitrary rational numbers (not just 0 and 1),
	\item step unbounded, i.e.\ not imposing any horizon for the temporal operators.
\end{itemize} 
Besides, we consider models with unbounded size, i.e.\ countable models or finite models, but with no a priori restriction on the size of the state space.  These are the three distinguishing features, compared to other works.  The closest are the following.  Firstly, solutions for the qPCTL have been given in \cite{HS,LICS} and for a more general logic PCTL$^*$ in \cite{LS,KL}.  Secondly, \cite{chakraborty2016satisfiability} shows decidability for \emph{bounded PCTL} where the scope of the operators is restricted by a step bound to a given time horizon.  Thirdly, the \emph{bounded satisfiability problem} is to determine, whether there exists a model of a given size for a given formula. This problem has been solved by encoding it into an SMT problem \cite{bertrand2012bounded}.  There is an important implication of this result.  Namely, if we are able to determine a maximum required model size for some formula, then it follows that the satisfiability of that formula can also be determined.  We take this approach in some of our proofs. Additionally, we use the result of \cite{LICS} that the branching degree (number of successors) for a model of a formula $\phi$ can be bounded by $|\phi|+2$, where $|\phi|$ is the length of $\phi$.

\textbf{Our contribution} is as follows:
\begin{itemize}
	\item We show decidability of the (finite and infinite) satisfiability problem for  several quantitative unbounded fragments of PCTL, focusing on future- and globally-operators (\F{},\G{}).
	\item We investigate the relationship between finite and infinite satisfiability on these fragments.
	\item We identify a fundamental issue preventing us from extending our techniques to the general case.
	We demonstrate this on a formula enforcing a more complicated form of its models.
	This allows us to identify the ``smallest elegant'' fragment where the problem remains open and the solution requires additional techniques.
\end{itemize}

Note that the considered fragments are not that interesting themselves. However, they illustrate the techniques that we developed and how far we can push decidability results when applying only those. Another fragment which might seem simple enough to be reasonable to consider is the pure \textbf{U}-fragment, but despite all efforts, we have not been able to show decidability for any interesting fragment thereof. For this reason, we will not consider general \textbf{U}-operators in this paper. Due to space constraints, the proofs are sketched and then worked out in detail in the Appendix.

\subsection{Further related work}

As for the \emph{non-probabilistic} predecessors of PCTL, the satisfiability problem is known to be EXPTIME-complete for CTL \cite{EH82} as well as the more general modal $\mu$-calculus \cite{BB:temp-logic-fixed-points,FL:PDL-regular-programs}.  Both logics have the small model property \cite{EH82,Kozen:mu-calculus-finite-model}, more precisely, every satisfiable formula $\phi$ has a finite-state model whose size is exponential in the size of $\phi$.  The complexity of the satisfiability problems has been investigated also for fragments of CTL \cite{KV:modular-model-checking} and the modal $\mu$-calculus \cite{HKM:univ-exis-mu-calculus-TCS}.

The satisfiability problem for qPCTL and qPCTL$^*$ was investigated already in the early 80’s \cite{LS,KL,HS}, together with the existence of sound and complete axiomatic systems.  The decidability for qPCTL over countable models also follows from these general results for qPCTL$^*$, but the complexity was not examined until \cite{LICS}, showing it is also EXPTIME-complete, both for finite and infinite satisfiability.

While the decidability of satisfiability is open, there are only few negative results. \cite{LICS} proves \emph{undecidability} of the problem whether for a given PCTL formula there exists a model with a branching degree that is bounded by a given integer, where the branching degree is the number of successors of a state. However, the authors have not been able to extend their proof and show the undecidability for the general problem.

The PCTL \emph{model checking problem} is the task to determine, whether a given system satisfies a given formula, i.e.\ whether it is a model of the formula.  This problem has been studied both for finite and infinite Markov chains and decision processes, see e.g. \cite{CY:probab-verification-JACM,HK:quantitative-mu-calculus-LICS,EY:RMC-SG-equations,EKM:prob-PDA-PCTL-LMCS,BKS:pPDA-temporal}.  The PCTL \emph{strategy synthesis} problem asks whether the non-determinism in a given Markov decision process can be resolved so that the resulting Markov chain satisfies the formula \cite{BGLBC:MDP-controller,KS:MDP-controller,BBFK:Games-PCTL-objectives,BFK:MDP-PECTL-objectives}.

\section{Preliminaries} % Main chapter title
\label{Chapter2}

In this section, we recall basic notions related to (discrete-time) Markov chains \cite{Norris:book} and the probabilistic CTL \cite{HJ:logic-time-probability-FAC}.  Let $\atomic$ be a finite set of atomic propositions.

\subsection{Markov chains}

\begin{definition}[Markov chain]
    A \emph{Markov chain} is a tuple $M = (S, P, \initstate, L)$ where $S$ is a countable set of \emph{states}, $P: S \times S \rightarrow [0,1]$ is the \emph{probability transition matrix} such that, for all $s \in S$, $\sum_{t \in S}{P(s,t)} = 1$, $\initstate\in S$ is the \emph{initial} state, and $L : S \rightarrow 2^{\atomic}$ is a labeling function.
\end{definition}
Whenever we write $M$, we implicitly mean a Markov chain $ (S, P, \initstate,L)$.  The semantics of a Markov chain $M$, is the probability space $(\runs_M,\mathcal F_M,\pr_M)$, where $\runs_M=S^\omega$ is the set of \emph{runs} of $M$, $\mathcal F_M\subseteq 2^{S^\omega}$ is the $\sigma$-algebra generated by the set of cylinders of the form $\mathit{Cyl}_M(\rho)=\{ \pi \in S^{\omega} \mid \rho \text{ is a prefix of } \pi \}$ and the probability measure is uniquely determined \cite{baier2008principles} by $\pr_M(\mathit{Cyl}_M(\rho_0\cdots\rho_n)) := \prod_{0 \leq i <n} P(\rho_i,\rho_{i+1})$ if $\rho_0=\initstate$ and $0$ otherwise.

We say that a state is \emph{reached} on a run if it appears in the sequence; a set of states is reached if some of its states are reached. The immediate successors of a state $s$ are denoted by $\mathit{post}_M(s) := {\{t \in S \mid P(s,t) > 0 \}}$ and the set of states reachable with positive probability is the reflexive and transitive closure $ post^*_M(s)$.  We will write $\pr(\cdot), \post(\cdot)$, and $\post^*(\cdot)$, if $M$ is clear from the context.

The \emph{unfolding} of a Markov chain $M$ is the Markov chain $T_M := (S^+, P',\initstate,L')$ with the form of an infinite tree given by $ P'(\rho s, \rho s s') = P(s,s')$ and $L'(\rho s)=L(s)$.  Each state of $T_M$ maps naturally to a state of $M$ (the last one in the sequence), inducing an equivalence relation $\rho s\sim\rho s'$ iff $s=s'$.  Consequently, each run of $T_M$ maps naturally to a run of $M$ and the unfolding preserves the measure of the respective events.

For a Markov chain $M$, a set $T \subseteq S$ is called \emph{strongly connected} if for all $s,t \in T$, $t \in \post^*(s)$; it is a \emph{strongly connected component (SCC)} if it is maximal (w.r.t. inclusion) with this property.  If, moreover, $\post^*(t)\subseteq T$ for all $t\in T$ then it is a \emph{bottom SCC (BSCC)}. A classical result, see e.g.\ \cite{baier2008principles}, states that the set of states visited infinitely often is almost surely, i.e.\ with probability 1, a BSCC:

\begin{lemma}
    \label{lem:bsccs}
    In every finite Markov chain, the set of BSCCs is reached almost surely.  Further, conditioning on runs reaching a BSCC $C$, every state of $C$ is reached infinitely often almost surely.
\end{lemma}

\subsection{Probabilistic Computational Tree Logic} % Main chapter title
\label{Chapter3}

%The interested reader shall refer to (\cite{baier2008principles}) for
%examples on PCTL, or in-depth discussions on LTL and CTL.  

The definition of probabilistic CTL (PCTL) \cite{HJ:logic-time-probability-FAC} is usually based on the next- and until-operators (\X{}, \U{}).  In this paper, we restrict our attention to the future- and globally operators (\F{}, \G{}), which can be derived from the until-operator.  Further, w.l.o.g.\ we impose the negation normal form and the lower-bound-comparison normal form; for the respective transformations see, e.g., \cite{baier2008principles}.

\begin{definition}[PCTL(\F{},\G{}) syntax and semantics]
    The \emph{formulae} are given by the following syntax:
    % of \emph{state formulae} and \emph{path    formulae}:
	%\begin{align*}
	\[
        \Phi ::= a \mid
                 \neg a \mid
                 \Phi \land \Phi \mid
                 \Phi \lor \Phi \mid
                 \F{\rhd q}\Phi \mid
                 \G{\rhd q}\Phi
                 %\PPP{\Psi}{\bowtie r}\\
%        \Psi ::= &         \F{\Phi} \mid        \G{\Phi}
\]
   % \end{align*}
    where $q \in [0,1]$, $\rhd \in \{ \geq, >\}$, and $a \in \atomic$ is an atomic proposition.  Let $M$ be a Markov chain and $s \in S$ its state.  We define the modeling relation $\models$ inductively as follows
    
    \begin{enumerate}[label=(M\arabic*),align=left]
    	\item \label{def:model.a} $M, s \models a$ iff $a \in L(s)$
    	\item \label{def:model.neg} $M, s \models \neg a$ iff $a \notin L(s)$
    	\item \label{def:model.and} $M, s \models \phi \land \psi$ iff
    	$M,s \models \phi$ and $M,s \models \psi$
    	\item \label{def:model.or} $M, s \models \phi \lor \psi$ iff
    	$M,s \models \phi$ or $M,s \models \psi$
    	\item \label{def:model.F} $M, s \models
    	\F{\rhd q}{\varphi}$ iff $\pr_{M(s)}(\{ \pi \mid
    	\exists i \in \mathbb{N}_0: M,\pi[i] \models \varphi \}) \rhd q$
    	\item \label{def:model.G} $M, s \models
    	\G{\rhd q}{\varphi}$ iff $\pr_{M(s)}(\{ \pi \mid
    	\forall i \in \mathbb{N}_0: M,\pi[i] \models \varphi \}) \rhd q$
    \end{enumerate}
    where $M(s)$ is $M$ with $s$ being the initial state, and $\pi[i]$ is the $i$th element of $\pi$. We say that $M$ is a \emph{model} of $\varphi$ if $M,\initstate \models \varphi$.
\end{definition}

We will denote the set of literals by $\literals := \atomic \cup \{ \neg a \mid a \in \atomic \}$. Instead of the constraint $\geq1$, we often write $=1$. Further, we define the set of all subformulae.  This definition slightly deviates from the usual definition of subformulae, e.g. the one in \cite{LICS}, in that $\neg a \in sub(\phi)$ does not necessarily imply $a \in sub(\phi)$. 

\begin{definition}[Subformulae]
The set $sub(\phi)$ is recursively defined as follows

\begin{itemize}
    \item $\phi \in sub(\phi)$
%    \item if $\neg \psi \in sub(\phi)$ and $\psi \not= a$,
%        then $\psi \in sub(\phi)$
    \item if $\psi \land \xi \in sub(\phi)$ or $\psi \lor \xi \in sub(\phi)$,
        then $\psi, \xi \in sub(\phi)$
    \item if $\F{\rhd q}{\psi} \in sub(\phi)$ or $\G{\rhd q}{\psi} \in sub(\phi)$
        then $\psi \in sub(\phi)$
\end{itemize}
\end{definition}

Next, we introduce the satisfiability problems, which are the main topic of the paper.

\begin{definition}[The satisfiability problems]
    A formula $\phi$ is called \emph{(finitely) satisfiable}, if there is a (finite) model for $\phi$.  Otherwise, it is  (finitely) unsatisfiable.  The (finite) satisfiability  problem is to determine whether a given formula is (finitely) satisfiable.
\end{definition}

Instead of simply writing ``satisfiable'' we sometimes stress the absence of ``finitely'' and write ``generally satisfiable'' for satisfiablity on countable, i.e. finite or countably infinite, models. For some proofs, it is more convenient to consider the unfolding of a Markov chain instead of the original one. As we mentioned already, the measure of events is preserved in the unfolding of a chain. Hence, we can state the following lemma.

\begin{lemma}\label{lem:tree}
	If $M$ is a model of $\phi$ then its unfolding $T_M$ is a model of $\phi$.
\end{lemma}

We say that formulae $\phi,\psi$ are (finitely) equivalent if they have the same set of (finite) models, written $\phi\equiv\psi$ ($\phi\equiv_{\mathit{fin}}\psi$); that they are (finitely) equisatisfiable if they are both (finitely) satisfiable or both (finitely) unsatisfiable; and that $\phi\Rightarrow\psi$ if every model of $\phi$ is also a model of $\psi$.

\section{Results}
In this section we present our results. A summary is schematically depicted in Fig.~\ref{fig:summary}.  We briefly describe the considered fragments; the full formal definitions can be found in the respective sections.  Since already the satisfiability for propositional logic in negation normal form has nontrivial instances only when all the constructs $a,\neg a$ and conjunction are present, we only consider fragments with all three included; see the bottom of the Hasse diagram.  The fragments are named by the list of constructs they use, where we omit the three constructs above to avoid clutter.  Here $1$ stands for $\geq1$ and $q$ stands for $\rhd q$ for all $q\in[0,1]\cap\mathbb Q$.  Further, $\G x(\mathit{list})$ denotes the sub-fragment of $\mathit{list}$ where the topmost operator is $\G x$.  Finally, $\F{q/1}$ denotes the use of $\F q$ with the restriction that inside $\G{}$ only $q=1$ can be used.  

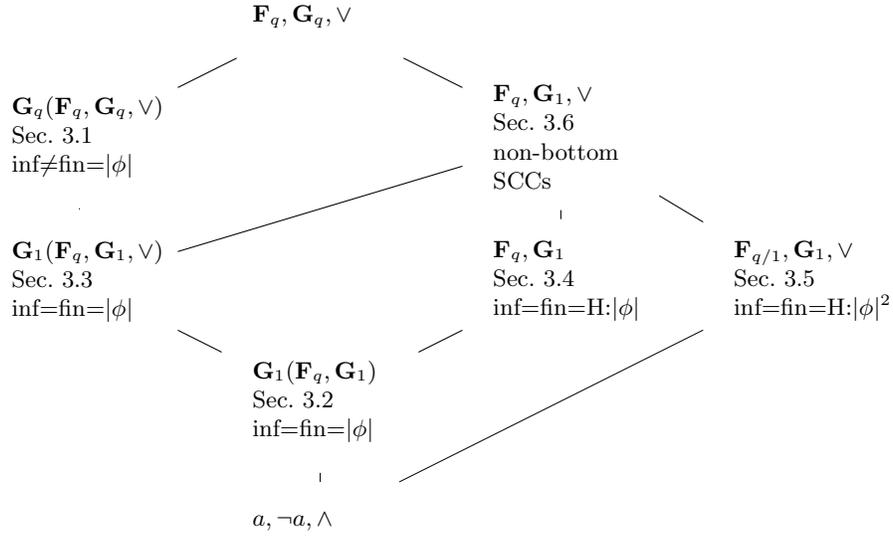
\begin{figure}
\begin{tikzpicture}[scale=0.8,text width=18mm,outer sep=3mm]
\node (gen) at (0,2) {$\F q,\G q,\vee$  };
\node (qual) at (4,0) {$\F q,\G 1,\vee$ Sec.~\ref{ss:qual} non-bottom SCCs};
\node (nodisj) at(4,-2.4) {$\F q,\G1$ Sec.~\ref{ss:nodisj} inf=fin=H:$|\phi|$};
\node (noqf) at(8,-2.4) {$\F{q/1},\G1,\vee$ Sec.~\ref{ss:noqf} inf=fin=H:$|\phi|^2$};
\node (g) at(-4,0) {$\G q(\F q,\G q,\vee)$  Sec.~\ref{ss:g} inf$\neq$fin=$|\phi|$};
\node (g1) at(-4,-2.4) {$\G1(\F q,\G 1,\vee)$ Sec.~\ref{ss:g1} inf=fin=$|\phi|$};
\node (base) at (0,-4.4) {$\G1(\F q,\G 1)$ Sec.~\ref{ss:base} inf=fin=$|\phi|$};
\node (bottom) at (0,-6.4) {$a,\neg a, \wedge$};
\draw 
(gen) edge (qual)
(gen) edge (g)
(qual) edge (nodisj)
(qual) edge (noqf)
(g) edge (g1)
(g1) edge (base)
(qual) edge (g1)
(nodisj) edge (base)
(base) edge[line width=0.01pt] (bottom)
(noqf) edge[line width=0.01pt] (bottom)
;
\end{tikzpicture}
    \caption{Hasse diagram summarizing the satisfiability results for the considered fragments of PCTL(\F{}, \G{}), all containing literals and conjunctions, and some form of quantitative comparisons. The fragments are described by the list of operators they allow (excluding the constructs of the minimal fragment). The subscript denotes the possible constraints on probabilistic operators. $\G{x}(list)$ denotes formulae in the fragment described by $list$ with \G{}-operators at the top-level. fin and inf abbreviate finite and general satisfiability, respectively. fin=inf denotes that the problems are equivalent. H:$x$ denotes that the height of a tree model can be bounded by $x$. By $=x$ we denote that the model size can be bounded by $x$. The $\F q, \G 1, \vee$-fragment might require non-bottom SCCs in finite models}
\label{fig:summary}
\end{figure}

The fragments are investigated in the respective sections.  We examine the problems of the general satisfiability (``inf'') and the finite satisfiability (``fin''); equality denotes the problems are equivalent.  We use two results to prove decidability of the problems.  Firstly, \cite{bertrand2012bounded} shows that given a formula $\phi$ and an integer $n$, one can determine whether or not there is a model for $\phi$ that has at most $n$ states.  Consequently, we obtain the decidability result whenever we establish an upper bound on the size of smallest models.  Here ``$|\phi|$'' denotes the satisfiability of a given $\phi$ on models of size $\leq |\phi|$.  Secondly, \cite{LICS} establishes that for any satisfiable PCTL formula there is a model with branching bounded by $|\phi|+2$.  Consequently, we obtain the decidability result whenever it is sufficient to consider trees of a certain height $H$ (with back edges) since the number of their nodes is then bounded by $(|\phi|+2)^H$. Here ``H:$n$'' denotes that the models can be limited to a height $H\leq n$. %The basis for other results is the behaviour of $\F{}$'s inside a $\G{}$.

While we obtain decidability in the lower part of the diagram, the upper part only treats finite satisfiability, and in particular for $\F q,\G1,\vee$, we only demonstrate that models with more complicated structure are necessary. Namely, the models may be of unbounded sizes for structurally same formulae---i.e. formulae which only differ in the constraints on the temporal operators---or require presence of non-bottom SCCs, see Section~\ref{ss:qual} and the discussion in Section~\ref{sec:disc}.

\subsection{Finite satisfiability for $\G q(\F q,\G q,\vee)$}\label{ss:g}

This section treats \G{}-formulae of the $\F q,\G q$-fragment, i.e.\ of PCTL$(\F{},\G{})$. In particular, it includes \G{>0}-formulae. In general, formulae in this fragment (even without quantified \F{} and \G{}-operators) can enforce rather complicated behaviour \cite{LICS}. Therefore, we will focus on finitely satisfiable formulae. We will see that they can be satisfied by rather simple models.

\begin{definition}
	$\G q(\F q,\G q,\vee)$-formulae are given by the grammar	
	\begin{align*}
	\Phi &::=\G{\rhd q}{\Psi}\\
	\Psi &::= a \mid
	\neg a \mid
	\Psi \land \Psi \mid
	\Psi \lor \Psi \mid
	\F{\rhd q}{\Psi} \mid
	\G{\rhd q}{\Psi}
	\end{align*}
\end{definition}

The main result of this section is that finitely satisfiable formulae in this fragment can be satisfied by models of size linear in $|\phi|$.

\begin{theorem}
	\label{thm:Gq(Fq,Gq,v)-size} Let $\phi$ be a finitely satisfiable $\G q(\F q,\G q,\vee)$-formula.  Then $\phi$ has a model of size at most $|\phi|$.
\end{theorem}

Intuitively, we obtain the result from the fact that some BSCC is reached almost surely and every state in a BSCC is reached almost surely, once we have entered one. In infinite models, BSCCs are not reached almost surely and therefore the proofs cannot be extended to general satisfiability. The following lemma and its proof demonstrate how we can make use of the BSCC properties in order to obtain an equisatisfiable formula in a simpler fragment.

\begin{lemma}
	\label{lem:Gq(Fq,Gq,v)-normal}
    Let $\phi$ be a $\G q(\F q,\G q,\vee)$-formula. Then, $\phi$ is finitely equisatisfiable to a $\G1(\F 1, \G 1)$-formula $\phi'$, such that $\phi' \Rightarrow \phi$.
\end{lemma}

\begin{proofsketch}
    Write $\phi$ as $\G{\rhd q}{\psi}$. Assume that we have a finite model $M$ for $\G{\rhd q}{\psi}$. Intuitively, we can select a BSCC that satisfies $\G{=1}{\psi}$. We know that there is a BSCC because we are dealing with a finite model. We also know that there is at least one BSCC satisfying our formula, for otherwise $M$ would not be a model for it. In a BSCC, every state is reached almost surely from every other state by Lemma~\ref{lem:bsccs}. Hence, we can select exactly one state for each \F{}-subformula which satisfies that formula's argument. Then we can create a new BSCC from these states, arranging them, e.g., in a circle. This BSCC models $\G{=1}{\hat{\psi}}$, where $\hat{\psi}$ replaces all probabilistic operators with their ``almost surely'' version.  Hence, we have created a model for a $\G1(\F 1, \G 1)$ formula from a model for $\G{\rhd q}{\psi}$. The opposite direction follows from the fact that $\G{=1}{\hat{\psi}} \Rightarrow \G{\rhd q}{\psi}$.
\QED\end{proofsketch}

Note that the transformation does not produce an equivalent formula. Hence, we cannot replace an occurrence of such a formula in a more complex formula. For instance, the formula $\G{\geq 1/2}\neg a \wedge \F{\geq 1/2}a$ is satisfiable, whereas $\G{=1}\neg a \wedge \F{\geq 1/2} a$ is not. The proof does not work for equality because we are selecting one BSCC while ignoring the rest. This example demonstrates why we cannot ignore certain BSCCs in general. Using the above result, it is easy to prove Theorem \ref{thm:Gq(Fq,Gq,v)-size}.

\begin{proofsketch}[Proof Sketch of Theorem \ref{thm:Gq(Fq,Gq,v)-size}]
    This follows immediately from the proof of Lemma \ref{lem:Gq(Fq,Gq,v)-normal}. The BSCC that we have created has at most as many states as there are \F{}-subformulae, which is bounded by $|\phi|$.
\QED\end{proofsketch}

\begin{example}
    Consider the formula 
    \begin{equation}
        \label{form:Gq(Fq,Gq)}
        \phi := \G{\geq 1/2}{(\F{\geq 1/3}{a} \land \F{\geq 1/3}{\neg a})}.
    \end{equation}

    The large Markov chain in Figure \ref{fig:ex-model-GqFq} models $\phi$. Unlabeled arcs indicate a uniform distribution over all successors. It is clear that the model is unnecessarily complicated. After reducing it according to Lemma~\ref{lem:Gq(Fq,Gq,v)-normal}, we obtain the smaller Markov chain on the right.  
    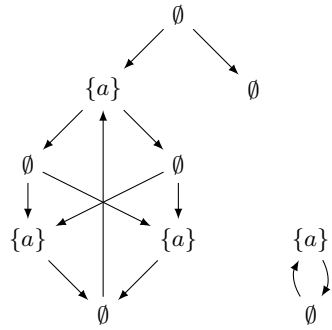
\begin{figure}
    	\centering
        \subfloat{
            \begin{tikzpicture}[bend angle=45]
                \node (s0) at (0,0) {$\emptyset$};
                \node (s1) at (-1,-1) {$\{a\}$};
                \node (s2) at (1,-1) {$\emptyset$};
                \node (s3) at (-2,-2) {$\emptyset$};
                \node (s4) at (0,-2) {$\emptyset$};
                \node (s5) at (-2,-3) {$\{a\}$};
                \node (s6) at (0,-3) {$\{a\}$};
                \node (s7) at (-1,-4) {$\emptyset$};
                \draw [->] (s0) to (s1);
                \draw [->] (s0) to (s2);
                \draw [->] (s1) to (s3);
                \draw [->] (s1) to (s4);
                \draw [->] (s3) to (s5);
                \draw [->] (s3) to (s6);
                \draw [->] (s4) to (s5);
                \draw [->] (s4) to (s6);
                \draw [->] (s5) to (s7);
                \draw [->] (s6) to (s7);
                \draw [->] (s7) to (s1);
            \end{tikzpicture}
        }
        \subfloat{
            \begin{tikzpicture}[auto,bend angle=30]
                \node (s0) at (0,0) {$\emptyset$};
                \node (s1) at (0,1) {$\{a\}$};
                \draw [->,bend left] (s0) to (s1);
                \draw [->,bend left] (s1) to (s0);
            \end{tikzpicture}
        }
        \caption{A large and a small model for Formula \eqref{form:Gq(Fq,Gq)}}
        \label{fig:ex-model-GqFq}
    \end{figure}
\end{example}

The example below shows that satisfiability is not equivalent to finite satisfiability for this fragment, and that the proposed transformation does not preserve equisatisfiability over general models.  The decidability of the general satisfiability thus remains open here.

\begin{example}
    Note that we made use of the BSCC properties for the proofs of this subsection, such as that some BSCC is reached almost surely. Since this is only the case for finite Markov chains, our transformation only holds for finite satisfiability.  If we consider the general satisfiability problem, then the equivalent of Lemma~\ref{lem:Gq(Fq,Gq,v)-normal} is not true. For instance, the formula
    \begin{equation}
        \label{form:infinite}
        \phi := \G{>0}{(\neg a \land \F{>0}{a} )}
    \end{equation}
    is satisfiable, but requires infinite models, as pointed out in \cite{LICS}. One such model is given in Figure \ref{fig:infinite-model}. Observe that the single horizontal run has measure greater than $0$. Now consider
    \[
        \hat{\phi} := \G{=1}{(\neg a\land\F{=1}{a} )}
    \]
    Obviously, this formula unsatisfiable. Hence, in this case $\phi$ is not
    equisatisfiable to $\hat{\phi}$.

    \begin{figure}
        \centering
        \begin{tikzpicture}[auto]
            \node (first top) at (0,2) {$\emptyset$};
            \node (second top) at (2,2) {$\emptyset$};
            \node (third top) at (4,2) {$\emptyset$};
            \node (fourth top) at (6,2) {$\emptyset$};
            \node (first bottom) at (0,0) {$\{a\}$};
            \node (second bottom) at (2,0) {$\{a\}$};
            \node (third bottom) at (4,0) {$\{a\}$};
            \node (fourth bottom) at (6,0) {$\{a\}$};
            \node (first phantom) at (7,2) {};
            \node (second phantom) at (8,2) {};
            \draw [->] (first top) to node {$1/2$} (second top);
            \draw [->] (second top) to node {$3/4$} (third top);
            \draw [->] (third top) to node {$7/8$} (fourth top);
            \draw [dotted] (fourth top) to (first phantom);
            \draw [->] (first top) to node {$1/2$} (first bottom);
            \draw [->] (second top) to node {$1/4$} (second bottom);
            \draw [->] (third top) to node {$1/8$} (third bottom);
            \draw [->] (fourth top) to node {$1/16$} (fourth bottom);
            \draw [->,loop right] (first bottom) to node {$1$} (first bottom);
            \draw [->,loop right] (second bottom) to node {$1$} (second bottom);
            \draw [->,loop right] (third bottom) to node {$1$} (third bottom);
            \draw [->,loop right] (fourth bottom) to node {$1$} (fourth bottom);
        \end{tikzpicture}
        \caption{An infinite model for Formula \eqref{form:infinite}}
        \label{fig:infinite-model}
    \end{figure}
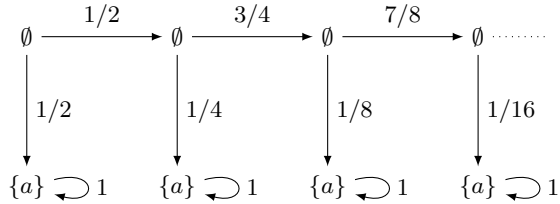
\end{example}

\subsection{Satisfiability for $\G1(\F q,\G 1)$}\label{ss:base}
\label{sec:G1(Fq,G1)}
This section treats \G{}-formulae of a fragment where $\G{\rhd q}$ only appears with $q=1$ and there is no disjunction.  The results are later utilized in a richer fragment in Section~\ref{sec:Fq,G1}. In fact, the main result of this section is an immediate consequence of the main theorem of Section~\ref{sec:G1(Fq,G1,v)}. Still, the results are interesting themselves as they show some properties of models for formulae in this fragment which do not apply in the generalized case.

\begin{definition}
	$\G1(\F q, \G 1)$-formulae are given by the grammar
    \begin{align*}
        \Phi &::= \G{=1}{\Psi}\\
        \Psi &::= a \mid
                 \neg a \mid
                 \Psi \land \Psi \mid
                 \F{\rhd q}{\Psi} \mid
                 \G{=1}{\Psi}
    \end{align*}
\end{definition}

We prove that satisfiable formulae of this fragment are satisfiable by models of linear size and thus also finitely satisfiable.

\begin{theorem}
	\label{thm:G1(Fq,G1)-size}
    Let $\phi$ be a satisfiable $\G1(\F q, \G 1)$-formula.  Then $\phi$ has a model of size at most $|\phi|$.
\end{theorem}

The idea here is that we can find a state which behaves similarly to a BSCC (even in infinite models) in that it satisfies all \G{}-subformulae. We can then use this state's successors to construct a small model. The outline of the proof is roughly as follows: First we show that from every state and for every subformula we can find a successor that satisfies this subformula. Using this, we can show that there is a state that satisfies all \G{}-subformulae.

\begin{lemma}
	\label{lem:G1(Fq,G1)-subformulae}
    Let $\phi$ be a satisfiable $\G1(\F q, \G 1)$-formula and $M$ its model. Then, for every $\psi \in \mathit{sub}(\phi)$, and $s \in S$, there is a state $t \in \mathit{post}^*(s)$, such that $M,t \models \psi$.
\end{lemma}

\begin{proofsketch}
    This follows from the fact that we do not allow disjunctions in this fragment. We apply induction over the depth of a subformula $\psi$. If the formula is $\phi$ itself, then there is nothing to show. Otherwise, the induction hypothesis yields that the higher-level subformulae are satisfied at some state $s$. From this, we can easily see that in all possible cases the claim follows: If the higher-level formula is a conjunction, then $\psi$ is one of its conjuncts. Since both conjuncts must be satisfied by $s$, in particular $\psi$ must be satisfied at $s$. A similar argument applies to \G{}-formulae. If it is of the form $\F{\rhd q}\xi$, then we know that there must be a reachable state where $\xi$ holds.
\QED\end{proofsketch}

This concludes the first part of the proof. We continue with the second part and prove that we can find a state which satisfies all \G{}-subformulae.

\begin{lemma}
	\label{lem:G1(Fq,G1)-allG} Let $\phi$ be a satisfiable $\G1(\F q, \G 1)$-formula, $M$ its model, and let $G := \{ \psi \in \mathit{sub}(\phi) \mid \psi = \G{=1}{\xi} \text{ for some } \xi \}$. Then there is a state $s \in S$ such that $M,s \models \psi$ for all $\psi \in G$.
\end{lemma}

\begin{proofsketch}
	It is clear that after encountering a \G{}-formula at some state, all successors will also satisfy it. Therefore, the set of satisfied \G{}-formulae is monotonically growing and bounded. Hence, we can apply induction over the number of yet unsatisfied \G{}-formulae.  In every step, we are looking for the next state to satisfy an additional \G{}-formula.  This is always possible (as long as there are still unsatisfied ones), due to Lemma \ref{lem:G1(Fq,G1)-subformulae}.
\QED\end{proofsketch}

Now, we can prove Theorem~\ref{thm:G1(Fq,G1)-size}.

\begin{proofsketch}[Proof Sketch of Theorem \ref{thm:G1(Fq,G1)-size}]
	By Lemma~\ref{lem:G1(Fq,G1)-allG}, we can find a state that satisfies all \G{}-subformulae. In some sense, this state's subtree resembles a BSCC. We can include exactly one state for each \F{}-subformula and create a BSCC out of those states, e.g., arrange them in a circle. We apply induction over $\psi \in sub(\phi)$. The satisfaction of literals and conjunctions is straightforward. Since every state is reached almost surely, every \F{}-formula will be satisfied that way. The satisfaction of the \G{}-formulae follows from the fact that all states used to satisfy all \G{}-formulae in the original model, and from the induction hypothesis.
\QED\end{proofsketch}

For the case of finite satisfiability, we also present an alternative proof, which sheds more light on this fragment and its super-fragments.  For details, see Appendix~\ref{app:proofs}.  Let $\equiv_{\mathit{fin}}$ denote equivalence of PCTL formulae over finite models.

\begin{theorem}
    \label{thm:G1(Fq,G1)-normal}
    Let $\phi$ be a $\G1(\F q, \G 1)$-formula. Then, the following equivalence holds:

    \begin{equation*}
        \G{=1}{\phi} \equiv_{\mathit{fin}}
        \G{=1}{(\bigwedge_{l \in A}{l} \land
        \F{=1}{\G{=1}{\bigwedge_{l \in B}{l}}} \land
        \bigwedge_{i \in I}\F{=1}{{\bigwedge_{l \in C_i}{l}}})}
    \end{equation*}

    for appropriate $I \subset \mathbb{N}$, and $A, B, C_i \subset \literals$.
\end{theorem}
\begin{proofsketch}
	The proof is based on the following auxiliary statements
	\begin{align}
        \G{=1}\G{=1}\phi &\equiv \G{=1}\phi\\
        \G{=1}\F {\rhd q}\phi &\equiv_{\mathit{fin}}\G{=1}\F{=1}\phi\label{eq:fin}\\
        \F{=1}\F{\rhd q}\phi &\equiv \F{\rhd q}\phi
	\end{align} 
	and follows by induction.
	
	The second statement is the most interesting one. Intuitively, it is a zero-one law, stating that infinitely repeating satisfaction with a positive probability ensures almost sure satisfaction.  Notably, this only holds if the probabilities are bounded from below, hence for finite models, not necessarily for infinite models.
\QED\end{proofsketch}

It is an easy corollary of this theorem that a satisfiable formula has a model of a circle form with $A$ and $B$ holding in each state and each element in each $C_i$ holding in some state.  In general the models can be of a lasso shape where the initial (transient) part only has to satisfy $A$, allowing for easy manipulation in extensions of this fragment.

\begin{remark}
	Note that the equivalence does not hold over infinite models. Indeed, consider as simple a formula as $\G{=1}\F{>0}a$, which is satisfied on the Markov chain of Fig.~\ref{fig:infinite-model} \cite{LICS}, while this does not satisfy the transformed $\G{=1}\F{=1}a$.  Crucially, equivalence (\ref{eq:fin}) does not hold  already for this tiny fragment. Interestingly, when we build a model for the transformed formula, which is equisatisfiable but not equivalent, it turns out to be a model of the original formula. If, moreover, we consider $\neg a, \land, \G{>0}$ then finite and general satisfiability start to differ.  
\end{remark}

Before we move on to the next fragment, we will prove another consequence of Lemma~\ref{lem:G1(Fq,G1)-subformulae}. It is a statement about the BSCCs of models for formulae in this fragment and will be used later for the proof of Theorem~\ref{thm:Fq,G1-size}.

\begin{corollary}
	\label{cor:G1(Fq,G1)-bsccs}
    Let $\G{=1}\phi$ be a satisfiable $\G1(\F q, \G 1)$-formula and $M$ its model. Then, for every BSCC $T \subseteq post^*(s_0)$ of $M$, the following holds
	
	\begin{enumerate}
		\item
		\label{cor:G1(Fq,G1)-bsccs.F}
            For all $\psi \in sub(\G{=1}\phi)$, there is a state $t \in T$, such that $M,t \models \psi$.
		\item
		\label{cor:G1(Fq,G1)-bsccs.G}
            For all $\G{=1}{\psi} \in sub(\G{=1}\phi)$, and for all states $t \in T$, $M,t \models \G{=1}{\psi}$.
	\end{enumerate}
\end{corollary}

\begin{proof}
	Point \ref{cor:G1(Fq,G1)-bsccs.F} follows from the fact that every reachable BSCC must satisfy $\G{=1}\phi$, and from Lemma \ref{lem:G1(Fq,G1)-subformulae}. Point \ref{cor:G1(Fq,G1)-bsccs.G} follows immediately from point \ref{cor:G1(Fq,G1)-bsccs.F}.
\QED\end{proof}

Note that we did not assume finite satisfiability here, so the model might not contain a single BSCC. In that case, the claim is trivially true. However, Theorem~\ref{thm:G1(Fq,G1)-size} allows us to focus on finitely satisfiable formulae in this fragment.

\subsection{Satisfiability for $\G1(\F q,\G1,\vee)$}
\label{sec:G1(Fq,G1,v)}\label{ss:g1}

This section treats \G{}-formulae of the fragment where $\G{\rhd q}$ only appears with $q=1$.  We thus lift a restriction of the previous fragment and allow for disjunctions. We generalize the obtained results to this larger fragment. We mentioned earlier that some of the results of the previous fragment do not apply here. Concretely, Lemma~\ref{lem:G1(Fq,G1)-subformulae} does not hold here; that is, there might be subformulae which are not satisfied almost surely. Therefore, there is not necessarily a state that satisfies all \G{}-subformulae. For example, consider $\G{=1}(\F{>0}\G{=1}a \vee \F{>0}\G{=1}\neg a)$. There cannot be a single BSCC to satisfy both disjuncts. Although this is not a problem for the results of this section, it will turn out to be a fundamental problem when dealing with arbitrary formulae of the $\F q,\G1,\vee$ fragment.

\begin{definition}
	$\G1(\F q, \G 1, \vee)$-formulae are given by the grammar
	\begin{align*}
	\Phi &::= \G{=1}{\Psi}\\
	\Psi &::= a \mid
	\neg a \mid
	\Psi \land \Psi \mid
	\Psi \lor \Psi \mid
	\F{\rhd q}{\Psi} \mid
	\G{=1}{\Psi}
	\end{align*}
\end{definition}

We prove that satisfiable formulae of this fragment are satisfiable by models of linear size and thus also finitely satisfiable.

%\todo{can the outer G be quantitative? would it affect equivalent vs. equisatisfiable?}

\begin{theorem}
	\label{thm:G1(Fq,G1,v)-size}
	Let $\phi$ be a satisfiable $\G1(\F q, \G 1, \vee)$-formula.  Then $\phi$ has a model of size at most $|\phi|$.
\end{theorem}

\begin{proofsketch}
	The proof for this theorem works essentially the same as it did for Theorem~\ref{thm:G1(Fq,G1)-size}. Recall that we looked for a state to satisfy all \G{}-formulae. Though we will not necessarily find a state that does so in this fragment, we can look for a state that satisfies maximal subsets of satisfied \G{}-formulae. Then, we can continue in a similar way as we did in the simpler setting.
\QED\end{proofsketch}

\subsection{Satisfiability for $\F q,\G1$}\label{ss:nodisj}
\label{sec:Fq,G1}

This section treats general formulae of the fragment with no disjunction and where $\G{\rhd q}$ only appears with $q=1$.

\begin{definition}
	$\F q, \G 1$-formulae are given by the grammar
	\begin{align*}
	\Phi ::= a \mid
	\neg a \mid
	\Phi \land \Phi \mid
	\F{\rhd q}{\Phi} \mid
	\G{=1}{\Phi}
	\end{align*}
\end{definition}

In Section~\ref{sec:G1(Fq,G1)} we discussed a special case of this fragment, where the top-level operator is $\G{=1}$. Two results are particularly interesting for this section: Firstly, the construction of models for such formulae as explained in the proof of Theorem~\ref{thm:G1(Fq,G1)-size}, and secondly, the properties of BSCCs in models for such formulae as stated in Corollary~\ref{cor:G1(Fq,G1)-bsccs}. We will use those in order to simplify models in this generalized setting. We say a Markov chain has \emph{height $h$} if it is a tree with back edges of height $h$.

\begin{theorem}
    \label{thm:Fq,G1-size}
    A satisfiable $\F q, \G 1$-formula $\phi$ has a model of height $|\phi|$.
\end{theorem}

\begin{proofsketch}
    Our aim is to transform a given, possibly infinite model into a tree-like shape. To do so, we first construct a tree by considering all non-nested \F{}-formulae. Each path in this tree will satisfy each of these \F{}-formulae at most once. At the end of each path, we will then insert BSCCs satisfying the \G{}-formulae, in the spirit of Theorem~\ref{thm:G1(Fq,G1)-size}.

    The collapsing procedure from a state $s$ is as follows: We first determine which of the \F{}-formulae that are satisfied at $s$ are relevant. Those are the formulae which are not nested in other temporal formulae and have not yet been satisfied on the current path. Once we have determined this set, say $I$, we need to find the successors which are required to satisfy the formulae in $I$. For this, we construct the set $sel(s)$. Informally, $sel(s)$ contains all states $t$ s.t. (i) $t$ satisfies at least one formula $\psi \in I$, and (ii) there is no state on the path between $s$ and $t$ satisfying $\psi$. Formally, $sel(s) := \{ t \in post^*(s) \mid \exists \F{\rhd r} \psi \in I.\ M,t \models \psi \land \forall t' \in post^*(s) \cap pre^*(t). M,t \not\models \psi \}$. Then, we connect $s$ to every state in $sel(s)$ directly; i.e. for $t \in sel(s)$, we set $P(s,t) := P^*(s,t)$, and for all states $t \not\in sel(s)$, we set $P(s,t) := 0$.\footnote{In fact, we need to scale $P(s,t)$ in order to obtain a Markov chain, in general, as the probability to reach $sel(s)$ might be less than $1$. For details, refer to the formal proof in the Appendix.} A simple induction on the length of a path yields that every state that is reachable from $s$ in the constructed MC is reached with at least the probability as in the original one. From this, one can easily see that every non-nested \F{}-formula is satisfied. The new set $post(s)$ might be infinite. However, we know that we can prune most of the successors and limit the branching degree to $|\phi|+2$ \cite{LICS}. Then, we repeat the procedure from each of the successors. Since the number of non-nested \F{}-formulae decreases with every step, we will reach states which do not have to satisfy non-nested \F{}-formulae at all on every branch. The number of steps we need to reach such states, is bounded by the number of non-nested \F{}-formulae in $\phi$. At those states, we can use Theorem~\ref{thm:G1(Fq,G1)-size} to obtain models for the respective \G{}-formulae. Those are of size linear in the size of the \G{}-formulae. The overall height is then bounded by $|\phi|$. The fact that the resulting MC is a model can be easily proved by induction over $|\phi|$.
\QED\end{proofsketch}

The models that we construct have a quite regular shape: They start as a tree and in every step ensure satisfaction of one of the \F{}-formulae. As soon as they have satisfied all outer \F{}-formulae, on every branch a model of circle shape for the respective \G{}-formula follows. Since the branching degree is at most $|\phi|+2$ and the number of steps before we repeat a state is bounded by $|\phi|$, the overall size is bounded to $(|\phi|+2)^{|\phi|}$.

\begin{example}
    \label{ex:selection-reduction}
    Let $\phi^p := \F{\geq p}{\G{=1}{a}}$, and $\psi^p := \F{\geq p}{\G{=1}{\neg a}}$. The large Markov chain of Figure \ref{fig:ex-red-Fq,G1} is a model for $\phi^{1/2} \land \psi^{1/2}$. The grayed states illustrate the set $sel(s_0)$. The other boxes show the sets $sel(.)$ of the respective grayed state. Everything in between is omitted. The smaller Markov chain is the reduced version of the original model.

    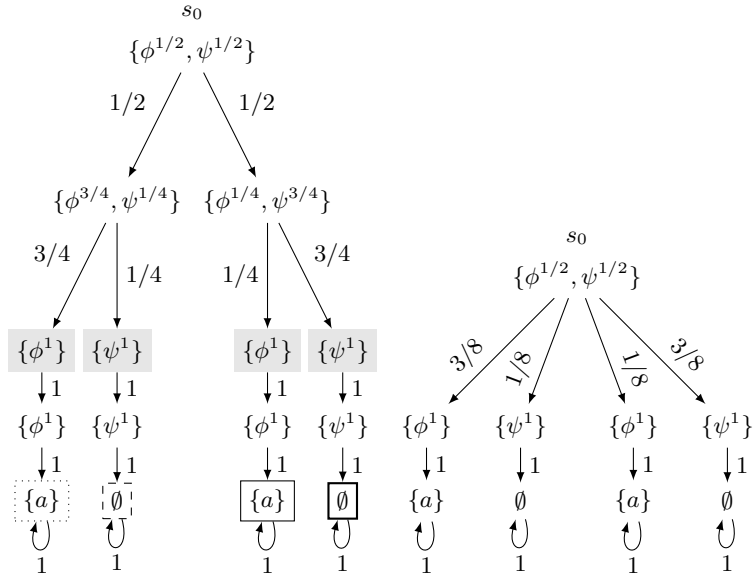
\begin{figure}
        \subfloat{
            \begin{tikzpicture}[auto]
                % states
                \node (s0) [label=above:$s_0$] at (0,0)
                    {$\{\phi^{1/2}, \psi^{1/2}\}$};
                \node (t1) at (-1,-2) {$\{\phi^{3/4}, \psi^{1/4}\}$};
                \node (t2) at (1,-2) {$\{\phi^{1/4}, \psi^{3/4}\}$};
                \node (s1) [fill=black!10,rectangle] at
                    (-2,-4) {$\{\phi^1\}$};
                \node (s2) [fill=black!10, rectangle] at (-1,-4)
                    {$\{\psi^1\}$};
                \node (s3) [fill=black!10,rectangle] at (1,-4)
                    {$\{\phi^1\}$};
                \node (s4) [fill=black!10,rectangle] at (2,-4)
                    {$\{\psi^1\}$};
                \node (t3) at (-2,-5) {$\{\phi^1\}$};
                \node (t4) at (-1,-5) {$\{\psi^1\}$};
                \node (t5) at (1,-5) {$\{\phi^1\}$};
                \node (t6) at (2,-5) {$\{\psi^1\}$};
                \node (t7) at (-2,-6) [draw=black,dotted,rectangle] {$\{a\}$};
                \node (t8) at (-1,-6) [draw=black,dashed,rectangle] {$\emptyset$};
                \node (t9) at (1,-6) [draw=black,solid,rectangle] {$\{a\}$};
                \node (t10) at (2,-6) [draw=black,thick,rectangle] {$\emptyset$};
                % arcs
                % level 0
                \draw [->] (s0) to node [swap] {$1/2$} (t1);
                \draw [->] (s0) to node {$1/2$} (t2);
                % level 1
                \draw [->] (t1) to node [swap] {$3/4$} (s1);
                \draw [->] (t1) to node {$1/4$} (s2);
                \draw [->] (t2) to node [swap] {$1/4$} (s3);
                \draw [->] (t2) to node {$3/4$} (s4);
                % level 2
                \draw [->] (s1) to node {$1$} (t3);
                \draw [->] (s2) to node {$1$} (t4);
                \draw [->] (s3) to node {$1$} (t5);
                \draw [->] (s4) to node {$1$} (t6);
                % level 3
                \draw [->] (t3) to node {$1$} (t7);
                \draw [->] (t4) to node {$1$} (t8);
                \draw [->] (t5) to node {$1$} (t9);
                \draw [->] (t6) to node {$1$} (t10);
                % level 4
                \draw [->,loop below] (t7) to node {$1$} (t7);
                \draw [->,loop below] (t8) to node {$1$} (t8);
                \draw [->,loop below] (t9) to node {$1$} (t9);
                \draw [->,loop below] (t10) to node {$1$} (t10);
            \end{tikzpicture}
        }
        \subfloat{
            \begin{tikzpicture}[auto]
                % states
                \node (s0) [label=above:$s_0$] at (0,0)
                    {$\{\phi^{1/2}, \psi^{1/2}\}$};
                \node (s1) at (-2,-2) {$\{\phi^1\}$};
                \node (s2) at (-0.75,-2) {$\{\psi^1\}$};
                \node (s3) at (0.75,-2) {$\{\phi^1\}$};
                \node (s4) at (2,-2) {$\{\psi^1\}$};
                \node (t1) at (-2,-3) {$\{a\}$};
                \node (t2) at (-0.75,-3) {$\emptyset$};
                \node (t3) at (0.75,-3) {$\{a\}$};
                \node (t4) at (2,-3) {$\emptyset$};
                % arcs
                % level 0
                \draw [->] (s0) to node [swap,sloped] {$3/8$} (s1);
                \draw [->] (s0) to node [swap,sloped] {$1/8$} (s2);
                \draw [->] (s0) to node [sloped] {$1/8$} (s3);
                \draw [->] (s0) to node [sloped] {$3/8$} (s4);
                % level 1
                \draw [->] (s1) to node {$1$} (t1);
                \draw [->] (s2) to node {$1$} (t2);
                \draw [->] (s3) to node {$1$} (t3);
                \draw [->] (s4) to node {$1$} (t4);
                % level 2
                \draw [->,loop below] (t1) to node {$1$} (t1);
                \draw [->,loop below] (t2) to node {$1$} (t2);
                \draw [->,loop below] (t3) to node {$1$} (t3);
                \draw [->,loop below] (t4) to node {$1$} (t4);
            \end{tikzpicture}
        }
        \caption{Example of a reduction for a $\F q, \G 1$-formula.  }
        \label{fig:ex-red-Fq,G1}
    \end{figure}
\end{example}

\subsection{Satisfiability for $\F{q/1},\G1,\vee$}\label{ss:noqf}
\label{sec:Fq1,G1,v}
In the previous section we have been able to construct simple models for formulae of the $\F q, \G1$-fragment by exploiting the nature of \G{}-formulae thereof as presented in Section~\ref{sec:G1(Fq,G1)}. This works because every formula nested within a \G{} is satisfied in every BSCC. Hence, we can simply postpone the satisfaction of those until we reach a BSCCs. In the $\F{q},\G1,\vee$-fragment, this is not the case anymore, as discussed in Section~\ref{sec:G1(Fq,G1,v)}. This can cause some complications, which are discussed in more detail in Section~\ref{sec:Fq,G1,v}. In order to be able to apply similar techniques as in the previous section, we can simplify the fragment and enforce the property that \F{}-formulae occur only with $q=1$ within \G{}s.

\begin{definition}
	$\F {q/1}, \G 1,\vee$-formulae are given by the grammar
\begin{align*}
    &\Phi ::= a \mid
             \neg a \mid
             \Phi \land \Phi \mid
             \Phi \lor \Phi \mid
             \F{\rhd q}{\Phi} \mid
             \G{=1}{\Psi} \\
    &\Psi ::= a \mid
             \neg a \mid
             \Psi \land \Psi \mid
             \Psi \lor \Psi \mid
             \F{=1}{\Psi} \mid
             \G{=1}{\Psi}
\end{align*}
\end{definition}

Again, we show that the necessary minimal height of models can be bounded.

\begin{theorem}
    \label{thm:Fq1,G1,v-size}
    A satisfiable $\F {q/1}, \G 1,\vee$-formula $\phi$ has a model of height $|\phi|^2$.
\end{theorem}

\begin{proofsketch}
    As a first step, we apply the same procedure as in the proof of Theorem~\ref{thm:Fq,G1-size}. The outer \F{}-formulae are then satisfied for the same reason as in the setting without disjunctions. However, the \G{}-nested \F{}-formulae might not be satisfied anymore because the BSCCs do not necessarily satisfy each of them. Since the \G{}-nested \F{}-formulae appear only with $q=1$, we know that once a state of the original model satisfies such a formula, almost every path satisfies the respective path formula. Let $s$ be a state of the reduced chain, and $t \in post(s)$. In the original model, there might be states in between $s$ and $t$. If some \G{}-nested \F{}-formula (say $\F{=1}\psi$) which is satisfied at $s$ is also satisfied at $t$, we do not need to take care of it. If this is not the case, then we know for sure that some of the states between $s$ and $t$ must satisfy $\psi$. We can determine such states for each \G{}-nested \F{}-formula. We include exactly one of those for each such formula. Then, preserving the order, we chain them in such a way that each one has a unique successor. The last one's unique successor is $t$. Let $s'$ be the first one. Then, we set $P(s,s') := P(s,t)$. We repeat this procedure for each state of the reduced chain.

    This way, we preserve the reachability probabilities and therefore the satisfaction of the outer \F{}-formulae. The newly added states guarantee the satisfaction of nested \F{}-formulae. An induction over $\phi$ shows that the constructed MC is again a model. Since we add at most $|\phi|$ new states between the states of the reduced MC which is of height at most $|\phi|$, we obtain the claimed bound on the height.
\QED\end{proofsketch}

    Figure \ref{fig:red-Fq,G1,v} illustrates the transformation of the models as described in the proof sketch. $sel_1$ and $sel_2$ are the selections. In the $\F q ,\G 1$-fragment, we directly connected those sets. Here, we insert simple chains between the selections.  The construction guarantees that we have at most one state per \F{}-formula to satisfy. This is obtained by postponing the satisfaction until the last possible moment before $sel$.

\begin{figure}
    \centering
    \begin{tikzpicture}
        [set/.style={rectangle, draw=black, fill=white},
         state/.style={circle, inner sep=0, minimum size=1mm,
         draw=black, fill=black}]
        \node[state, label=above:$s_0$] (s0) at (0,0) {};
        \node[state] (s1) at (-0.5,-0.5) {};
        \node[state] (s2) at (-0.5,-0.75) {};
        \node[state] (s3) at (-1,-2.25) {};
        \node[state] (s4) at (-1,-2.5) {};
        \node[state] (t1) at (0.5,-0.5) {};
        \node[state] (t2) at (0.5,-0.75) {};
        \node[state] (t3) at (1,-2.25) {};
        \node[state] (t4) at (1,-2.5) {};
        \node[set,minimum width=10mm] (sel1) at (0,-1.5) {$sel_1$};
        \node[set,minimum width=20mm] (sel2) at (0,-3.25) {$sel_2$};
        \node[set,minimum width=20mm,minimum height=10mm] (BSCCs)
            at (0,-4) {BSCCs};
        \draw[->] (s0) -- (s1);
        \draw[->] (s1) -- (s2);
        \draw[dotted] (s2) -- (sel1.north west);
        \draw[dotted] (-0.25,-0.75) -- (0.25,-0.75);
        \draw[->] (s0) -- (t1);
        \draw[->] (t1) -- (t2);
        \draw[dotted] (t2) -- (sel1.north east);
        \draw[->] (sel1.south west) -- (s3);
        \draw[->] (s3) -- (s4);
        \draw[dotted] (s4) -- (sel2.north west);
        \draw[dotted] (-0.5,-2.5) -- (0.5,-2.5);
        \draw[->] (sel1.south east) -- (t3);
        \draw[->] (t3) -- (t4);
        \draw[dotted] (t4) -- (sel2.north east);
    \end{tikzpicture}
    \caption{Reduction of models for $\F {q/1} ,\G 1, \vee$}
    \label{fig:red-Fq,G1,v}
\end{figure}
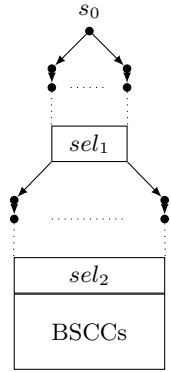

\begin{remark}
    Note that the resulting models have a tree shape with BSCCs. There are no non-bottom SCCs. Even if the original model did have such, they are removed by this construction: The reduction algorithm takes care that every non-nested \F{}-formula occurs at most once on every path. The inserted chains contain at most one state per nested \F{}-formula and do not introduce cycles.
\end{remark}

\begin{example}\label{ex:Fq1}
    Consider the formula

    $$\phi := \F{\geq 1/2}{(\G{=1}{a})} \land \G{=1}{(\F{=1}{\neg a} \lor \F{=1}{b})}.$$

    Figure \ref{fig:ex-Fq,G1,v-reduction} (a) shows a model for $\phi$.  The boxes illustrate the selection of $s_0$.  Figure \ref{fig:ex-Fq,G1,v-reduction} (b) shows the corresponding reduced chain. However, it is not a model for $\phi$. The reason is that neither states satisfying $\neg a$ nor such that satisfy $b$ are reached almost surely from $s_0$. By including additional states, the chain in Figure \ref{fig:ex-Fq,G1,v-reduction} (c) corrects this, and thereby we obtain a model of $\phi$.

    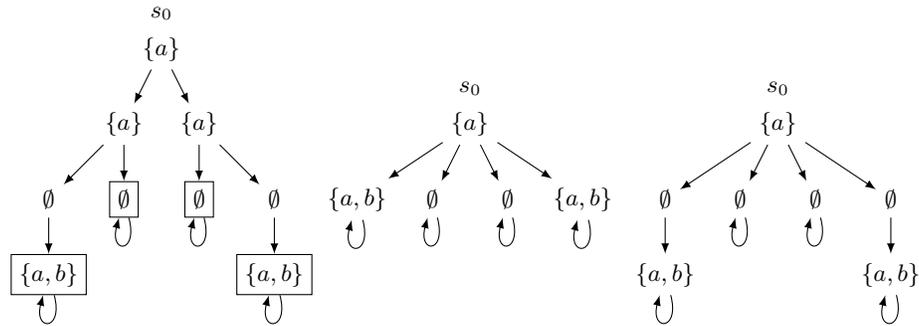
\begin{figure}
        \begin{tikzpicture}[auto]
            % states
            \node (s0) [label=above:$s_0$] at (0,0) {$\{a\}$};
            \node (s1) at (-0.5,-1) {$\{a\}$};
            \node (s2) at (0.5,-1) {$\{a\}$};
            \node (s3) at (-1.5,-2) {$\emptyset$};
            \node (s4) [rectangle,draw=black] at (-0.5,-2) {$\emptyset$};
            \node (s5) [rectangle,draw=black] at (0.5,-2) {$\emptyset$};
            \node (s6) at (1.5,-2) {$\emptyset$};
            \node (s7) [rectangle,draw=black] at (-1.5,-3) {$\{a,b\}$};
            \node (s8) [rectangle,draw=black] at (1.5,-3) {$\{a,b\}$};
            % selections
            %\node [color=color1] at (2,0) {$sel_M(s_0)$};
            % arcs
            % level 0
            \draw [->] (s0) to (s1);
            \draw [->] (s0) to (s2);
            % level 1
            \draw [->] (s1) to (s3);
            \draw [->] (s1) to (s4);
            \draw [->] (s2) to (s5);
            \draw [->] (s2) to (s6);
            % level 2
            \draw [->] (s3) to (s7);
            \draw [->] (s6) to (s8);
            % level 3
            \draw [->,loop below] (s7) to (s7);
            \draw [->,loop below] (s8) to (s8);
            \draw [->,loop below] (s4) to (s4);
            \draw [->,loop below] (s5) to (s5);
        \end{tikzpicture}
        \begin{tikzpicture}[auto]
            % states
            \node (s0) [label=above:$s_0$] at (0,0) {$\{a\}$};
            \node (s1) at (-1.5,-1) {$\{a,b\}$};
            \node (s2) at (-0.5,-1) {$\emptyset$};
            \node (s3) at (0.5,-1) {$\emptyset$};
            \node (s4) at (1.5,-1) {$\{a,b\}$};
            \node[white]  (x) at (1.5,-2) {$\{a,b\}$};
            % arcs
            % level 0
            \draw [->] (s0) to (s1);
            \draw [->] (s0) to (s2);
            \draw [->] (s0) to (s3);
            \draw [->] (s0) to (s4);
            % level 1
            \draw [->,loop below] (s1) to (s1);
            \draw [->,loop below] (s2) to (s2);
            \draw [->,loop below] (s3) to (s3);
            \draw [->,loop below] (s4) to (s4);
            \draw [->,loop below,white] (x) to (x);
        \end{tikzpicture}
        \begin{tikzpicture}[auto]
            % states
            \node (s0) [label=above:$s_0$] at (0,0) {$\{a\}$};
            \node (s1) at (-1.5,-1) {$\emptyset$};
            \node (s2) at (-0.5,-1) {$\emptyset$};
            \node (s3) at (0.5,-1) {$\emptyset$};
            \node (s4) at (1.5,-1) {$\emptyset$};
            \node (s5) at (-1.5,-2) {$\{a,b\}$};
            \node (s6) at (1.5,-2) {$\{a,b\}$};
            % arcs
            % level 0
            \draw [->] (s0) to (s1);
            \draw [->] (s0) to (s2);
            \draw [->] (s0) to (s3);
            \draw [->] (s0) to (s4);
            % level 1
            \draw [->] (s1) to (s5);
            \draw [->] (s4) to (s6);
            % level 2
            \draw [->,loop below] (s3) to (s3);
            \draw [->,loop below] (s2) to (s2);
            \draw [->,loop below] (s5) to (s5);
            \draw [->,loop below] (s6) to (s6);
        \end{tikzpicture}
        \caption{Example of a reduction in the $\F {q/1},\G 1, \vee$ fragment: (a) original model, (b) reduced model, (c) corrected model}
        \label{fig:ex-Fq,G1,v-reduction}
    \end{figure}
\end{example}

\subsection{Finite Models for $\F q, \G 1, \vee$}\label{ss:qual}
\label{sec:Fq,G1,v}
In this section, we discuss PCTL$(\F{},\G{})$ where $\G{}$ appears only with the contraint $q=1$.  Previously for the $\F q ,\G 1$-fragment, i.e. without disjunctions, we started from a model which was unfolded for a number of steps; we simplified such a model by dropping states (including all non-bottom SCCs) and then we inserted simple chains that guarantee the satisfaction of the nested \F{}-formulae. The resulting model thus (i) does not contain any non-bottom SCCs and (ii) the size only depends on the structure of the formula, not on the constraints of the \F{}-formulae. However, in the general $\F q,\G 1,\vee$-fragment, we cannot insert such simple chains to satisfy nested \F{}s.  Instead, we may have to branch at several places. Intuitively, the reason for such complications is the presence of a \emph{repeated, controlled choice}. This enables us to find a formula which requires more complicated models, namely models which either have non-bottom SCCs or are of size which also depends on the constraints and not only on the structure of the formula, or can even be infinite.

\begin{example}\label{ex:rep-choice}
    Consider $\phi := \G{=1} (\F{=1}(a \wedge \F{>0}\neg a) \vee a) \wedge \F{=1}\G{=1}a \wedge \neg a$. We can try to construct a model for $\phi$ as follows: Firstly, we have to start at a state which satisfies $\neg a$. This enforces the satisfaction of the first disjunct in the \G{}-formula. Therefore, almost all paths must lead to a state satisfying $a \land \F{>0}\neg a$. This state must eventually reach a state that satisfies $\neg a$ again with positive probability. Hence, we find ourselfs in the same situation as was the case in the initial state. So, we need to either create an SCC of alternating states that satisfy $a$ and $\neg a$, or create an infinite model. If we create an SCC, the side constraint $\F{=1}\G{=1}a$ enforces us to eventually leave this SCC. Hence, it is a non-bottom SCC. The MC in Figure~\ref{fig:mod-Fq,G1,v} is a possible model. From $s_0$ it models $\phi$, for any $p \in (0,1)$.

    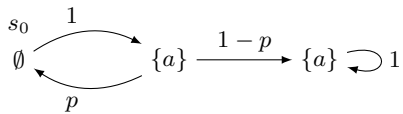
\begin{figure}
        \centering
        \begin{tikzpicture}[auto]
            \node (s0) [label=above:$s_0$] at (0,0) {$\emptyset$};
            \node (s1) at (2,0) {$\{a\}$};
            \node (s2) at (4,0) {$\{a\}$};
            \draw[->,bend left] (s0) to node {$1$} (s1);
            \draw[->,bend left] (s1) to node {$p$} (s0);
            \draw[->] (s1) to node {$1-p$} (s2);
            \draw[->,loop right] (s2) to node {$1$} (s2);
        \end{tikzpicture}
        \caption{Example model for $\phi$}
        \label{fig:mod-Fq,G1,v}
    \end{figure}
\end{example}

Note that the formula given in the example is qualitative. For this fragment, it is known how to solve the satisfiability problem already \cite{LICS}. However, we can easily adapt the formula to be quantitative. In this case, we might still be able to obtain a model for the quantitative version by keeping its shape and only adapting the probabilties of the model for the qualitative version. The question arises, whether this is possible in general or not. This question remains open and might be interesting for future work.

\section{Discussion, Conclusion, and Future Work}
\label{Chapter6}\label{sec:disc}

We have identified the pattern of the \emph{controlled repeated choice}, i.e. formulae of the form

$$\G{=1}(\phi_1 \vee \cdots \vee \phi_n)$$

where at least one of the $\phi_i$ contains an \F{}-formula that has a constraint other than $q=1$. Additionally, we have ``controlling'' side constraints, as in Example~\ref{ex:rep-choice}. We have seen that the presence of this pattern enforces more complicated structure of models even in the qualitative setting. This pattern is expressible in the $\F q,\G 1,\vee$-fragment.  Whenever we

\begin{enumerate}[label=(\alph*)]
	\item drop the side constraints, keeping only the $\G{=1}$-part, i.e. consider the $\G1(\F q,\G 1,\vee)$-fragment, or
	\item drop the disjunction for the choice and consider the $\F q,\G1$-fragment, or
	\item drop the quantity of the choice and consider the $\F{q/1},\G1,\vee$-fragment,
\end{enumerate}
the structure is simpler and we obtain decidability. For these fragments, we have even shown that the general satisfiability problem is equivalent to the finite satisfiability problem.
 
Further, adding quantities to $\G{}$-constraints obviously also makes the satisfiability problem more complicated.  Already for the qualitative $\G{>0}(\F{>0})$-fragment, satisfiability and finite satisfiability differ.  Nevertheless, we established the decidability of finite satisfiability even for the $\G q(\F q,\G q,\vee)$-fragment. 
  
\todo{bf-example, complexity}

\bigskip

Consequently, instead of attacking the whole quantitative PCTL or even just PCTL($\F{},\G{}$), we suggest two easier tasks, which should lead to a fundamental increase of understanding the general problem, namely:
\begin{itemize}
	\item finite (and also general) satisfiability of the $\F q,\G 1,\vee$-fragment, i.e. PCTL($\F{},\G{}$) where $\G{}$ is limited to the $=1$ constraint, and
	\item infinite satisfiability of the $\G q(\F q,\G q,\vee)$-fragment, i.e. $\G{}$-formulae of PCTL($\F{},\G{}$). 
\end{itemize}
While the former omits issues stemming from $\G{>0}$ \cite{LICS} and only deals with the repeated choice, the latter generalizes the qualitative results for $\G{>0},\G{=1}$ \cite{HS,LICS} in the presence of general quantitative $\F{}$'s.

Further, potentially more straight-forward, directions include the generalization of the results obtained in this paper to the until- and release-operators instead of future- and globally-operators, respectively, or the introduction of the next-operator.

%\iftrue
%\iffalse % for shorter compilation

%\newpage

\bibliographystyle{plainurl}% the recommnded bibstyle
\bibliography{bibliography}

\newpage
\appendix

\section{Full Proofs}
\label{app:proofs}
Whenever we write $M$ we implicitly mean a Markov chain $ (S, P, \initstate,L)$ and similarly $M'$ stands for $ (S', P', \initstate',L')$, and so on. Moreover, we will frequently refer to the formulae satisfied at $s$ by $sat(s) := \{ \psi \in sub(\phi) \mid M,s \models \psi \}$.

\subsection{Finite satisfiability for $\G q(\F q,\G q,\vee)$}

\begin{reflemma}{lem:Gq(Fq,Gq,v)-normal}
    Let $\phi$ be a $\G q(\F q,\G q,\vee)$-formula. Then, $\phi$ is finitely equisatisfiable to a $\G1(\F 1, \G 1, \vee)$-formula $\phi'$, such that $\phi' \Rightarrow \phi$.
\end{reflemma}

\begin{proof}
    Write $\phi := \G{\rhd r}{\psi}$. Let $M$ be a finite model for $\G{\rhd r}{\psi}$. Then, there must be at least one BSCC $T$, and a state $t \in T$, such that $M,t \models \G{\rhd' r'}{\psi}$. For a formula $\xi$, we define $\hat{\xi}$ recursively as follows

    $$
    \hat{\xi} := \begin{cases}
        a &\text{ if } \xi = a\\
        \hat{\zeta} \land \hat{\vartheta} &\text{ if } \xi = \zeta \land
        \vartheta\\
        \hat{\zeta} \lor \hat{\vartheta} &\text{ if } \xi = \zeta \lor
        \vartheta\\
        \F{=1}{\hat{\zeta}} &\text{ if } \xi = \F{\rhd r}{\zeta}\\
        \G{=1}{\hat{\zeta}} &\text{ if } \xi = \G{\rhd r}{\zeta}.
    \end{cases}
    $$

    Let $\xi \in sub(\psi)$, $t \in T$, and $M,t \models \xi$. We will show that $M,t \models \hat{\xi}$.

    \begin{enumerate}[label=\Roman*]
        \item $\xi = a$. Then, $\hat{\xi} = a = \xi$, and thus there is nothing
            to show.

        \item $\xi = \zeta \land \vartheta$.  Then, $M,t \models \zeta$ and $M,t \models \vartheta$. By the induction hypothesis, $M,t \models \hat{\zeta}$ and $M,t \models \hat{\vartheta}$, and thus $M,t \models \hat{\zeta} \land \hat{\vartheta}$.

        \item $\xi = \zeta \lor \vartheta$.  Then, $M,t \models \zeta$ or $M,t \models \vartheta$. By the induction hypothesis, $M,t \models \hat{\zeta}$ or $M,t \models \hat{\vartheta}$. Thus, $M,t \models \hat{\zeta} \lor \hat{\vartheta}$.

        \item $\xi = \F{\rhd r}{\zeta}$. Then, there is a state $t' \in T$, such that $M,t' \models \zeta$. By the induction hypothesis, $M,t' \models \hat{\zeta}$. Since $T$ is a BSCC, $t'$ is reached almost surely. Therefore, $M,t \models \F{=1}{\hat{\zeta}}$.

        \item $\xi = \G{\rhd r}{\zeta}$. Assume there was a state $t' \in T$, such that $M,t' \not\models \zeta$. Since $T$ is a BSCC, $t'$ is reached almost surely, and therefore $M,t \not\models \G{\rhd r}{\zeta}$, which is a contradiction. Hence, $M,t' \models \zeta$, for all $t' \in T$. By the induction hypothesis, $M,t' \models \hat{\zeta}$, for all $t' \in T$. Finally, this implies that $M,t \models \G{=1}{\hat{\zeta}}$.
    \end{enumerate}

    We have now shown that if $\phi$ is satisfiable, then so is $\hat{\phi}$.  Moreover, it is obvious that $\hat{\phi} \Rightarrow \phi$. The other direction follows immediately from this fact. Hence, $\phi$ is equisatisfiable to $\hat{\phi}$.
\QED\end{proof}

\begin{reftheorem}{thm:Gq(Fq,Gq,v)-size}
    Let $\phi := \G{\rhd r}{\psi}$ be a finitely satisfiable $\G q(\F q,\G q,\vee)$ formula. Then, there is a model of size linear in $|\phi|$.
\end{reftheorem}

\begin{proof}
    By Lemma~\ref{lem:Gq(Fq,Gq,v)-normal} we can consider $\hat{\phi}$ instead of $\phi$. Let $M$ be a model for $\hat{\phi}$. Let $S/\xi := \{ s \in S \mid M,s \models \xi \}$ denote all states that satisfy $\xi$. Then, we define $M'$, such that $S' := \{ S/\xi \mid \xi \in sub(\hat{\phi}) \text{ and } S/\xi \not= \emptyset \}$, $L'(T) := \bigcap_{t \in T}(L(t))$, and $P'$ an arbitrary distribution that generates a BSCC from $S'$---e.g. a circle. Note that $|S'| \leq |sub(\hat{\phi})|$. In particular it is always finite, even if $S$ is infinite.
    
    Let $\xi \in sub(\hat{\phi})$ and $T \in S'$, such that for all $t \in T$, $M,t \models \xi$. Then we will prove that $M',T \models \xi$ by induction over $\xi$.

    \begin{enumerate}[label=\Roman*]
        \item $\xi = a$ or $\xi = \neg a$. Assume $\xi = a$. Then, for all $t \in T$, $a \in L(t)$. Hence, $a \in L'(T)$ and $M',T \models a$. The case $\xi = \neg a$ is analagous.

        \item $\xi = \zeta \land \vartheta$. Then, for all $t \in T$, $M,t \models \zeta$ and $M,t \models \vartheta$. By the induction hypothesis, $M',T \models \zeta$ and $M',T \models \vartheta$ and therefore $M',T \models \zeta \land \vartheta$.
            
        \item $\xi = \zeta \lor \vartheta$. Then, for all $t \in T$, $M,t \models \zeta$ or $M,t \models \vartheta$. By the induction hypothesis, $M',T \models \zeta$ or $M',T \models \vartheta$. Hence, $M',T \models \zeta \lor \vartheta$.

        \item $\xi = \F{=1}{\zeta}$. Since all $t \in T$ model $\F{=1}\zeta$ (and $T$ is non-empty), there must be a state $s \in S$, such that $M,s \models \zeta$. Then, $S/\zeta \in S'$ and, by definition, for all $s \in S/\zeta$, $M,s \models \zeta$. Therefore, by the induction hypothesis, $M',S/\zeta \models \zeta$. Since $M'$ is a single BSCC, $S/\zeta$ is reached almost surely. Thus, $M,T \models \F{=1}\zeta$.
            
        \item $\xi = \G{=1}{\zeta}$. It must hold for all $s \in S$, that $M,s \models \zeta$. In particular this means that for all $T' \in S'$ and all $t \in T'$, $M,t \models \zeta$. By the induction hypothesis, $M',T' \models \zeta$. Finally, this implies that $M',T' \models \G{=1}\zeta$ and then also $M',T \models \G{=1}\zeta$.\QED
    \end{enumerate}
\end{proof}

\subsection{Satisfiability for $\G1(\F q,\G 1)$}%\label{ss:base}

For two formulae $\phi, \psi$, let $\psi \prec \phi$ iff $\psi \in sub(\phi)$ and $\psi \not= \phi$.  Moreover, for a set of formulae $\Phi$, we denote

\begin{align*}
    \mathfrak{P}_n := &\{ \psi_1 \dots \psi_n \mid \text{for all i } \in \{1,\dots ,n-1 \}.\\
        &(\psi_{i+1} \prec \psi_i \text{ and there is no } \xi.(\psi_{i+1} \prec \xi \prec \psi_{i}))\}\\
    \mathfrak{P} := &\bigcup_{n \in \mathbb{N}}{\mathfrak{P}_n}.
\end{align*}

\begin{reflemma}{lem:G1(Fq,G1)-subformulae}
    Let $\phi$ be a $\F q,\G 1$ formula, and $M$ a model for $\G{=1}{\phi}$.  Then, for every $\psi \in sub(\phi)$, and $s \in S$, there is a state $t \in post^*(s)$, such that $M,t \models \psi$.
\end{reflemma}

\begin{proof}
    Let $\psi_1 \dots \psi_n \in \mathfrak{P}$, such that $\psi_1=\phi$ and $\psi_n=\psi$. We apply induction over $n$.

    \begin{enumerate}[label=\Roman*]
        \item $n=1$. Then, $\phi=\psi$ and thus for all $t \in post^*(s_0)$, $M,t \models \G{=1}{\psi}$. Therefore $M,t \models \psi$.

        \item $n=n'+1$. By the induction hypothesis, there is a state $t \in post^*(s)$, such that $M,t \models \psi_{n'}$. We have to show that there is a state $t' \in post^*(s)$, with $M,t' \models \psi_n$.  Consider the following cases:

            \begin{enumerate}
                \item $\psi_{n'}=\psi_n \land \xi$. Then, $M,t \models \psi_n$.

                \item $\psi_{n'}=\G{=1}{\psi_n}$. Then, $M,t \models \psi_n$.

                \item $\psi_{n'}=\F{\rhd r}{\psi_n}$. Then, there must be a state $t' \in post^*(t)$, such that $M,t' \models \psi_n$.\QED
            \end{enumerate}
    \end{enumerate}
\end{proof}

\begin{reflemma}{lem:G1(Fq,G1)-allG}
    Let $\phi$ be a $\F q,\G 1$ formula, and $M$ a model for $\G{=1}{\phi}$.  Moreover, let $G := \{ \psi \in sub(\phi) \mid \psi = \PG{\xi}{=1} \text{ for some } \xi \}$. Then, there is a state, $s \in S$, such that for all $\psi \in G$, $M,s \models \psi$.
\end{reflemma}

\begin{proof}
    Let $s \in S$. A straightforward induction over $n:= |G \setminus sat(s)|$ yields the claim.

    \begin{enumerate}[label=\Roman*]
        \item $n=0$. Then, we are done.

        \item $n=n'+1$. Let $\psi \in G \setminus sat(s)$. Due to Lemma~\ref{lem:G1(Fq,G1)-subformulae}, there is a state $t \in post^*(s)$, with $M,t \models \psi$. Since all formulae in $G$ are \G{}-formulae, $G \setminus sat(t) \subset G \setminus sat(s)$, hence $|G \setminus sat(t)| < |G \setminus sat(s)|$. Now, the claim follows from the induction hypothesis.\QED
    \end{enumerate}
\end{proof}

\begin{refcorollary}{cor:G1(Fq,G1)-bsccs}
    Let $\phi$ be a $\F q,\G 1$ formula, and $M$ a finite model for $\G{=1}{\phi}$.  Then, for every BSCC $T \subseteq post^*(s_0)$, the following holds

	\begin{enumerate}
		\item \label{cor:G1(Fq,G1)-bsccs.F} For all $\psi \in sub(\phi)$, there is a state $t \in T$, such that $M,t \models \psi$.
		\item \label{cor:G1(Fq,G1)-bsccs.G} For all $\G{=1}{\psi} \in sub(\phi)$, and for all states $t \in T$, $M,t \models \G{=1}{\psi}$.
	\end{enumerate}
\end{refcorollary}

\begin{proof}
    Let $\psi \in sub(\phi)$, $T \subseteq S$ be a BSCC, and $t \in T$. Lemma~\ref{lem:G1(Fq,G1)-subformulae} states that there is a $t' \in post^*(t) = T$, such that $M,t' \models \psi$. If $\psi = \G{=1}{\xi}$, then it is clear that for all $t' \in T$, $M,t' \models \psi$.
\QED\end{proof}

\begin{reftheorem}{thm:G1(Fq,G1)-size}
    Let $\phi := \G{=1}\psi$, and $M$ a model for $\phi$. Then, there is a finite model for $\phi$.
\end{reftheorem}

\begin{proof}
    Lemma~\ref{lem:G1(Fq,G1)-allG} yields a state $s \in S$ which satisfies all \G{}-subformulae of $\phi$. Let $S/\xi := \{ s \in post^*(s) \mid M,s \models \xi \}$. Then, define $M'$, such that $S' := \{ S/\xi \mid \xi \in sub(\phi) \}$, $L'(T) := \bigcap_{t \in T}(L(t))$, and $P'$ generating an arbitrary BSCC of $S'$---e.g. a circle. Observe that $|S'| = |sub(\phi)|$, hence in particular finite, even though the individual $S/\xi$ might be infinite.

    Now, let $\xi \in sub(\phi)$, and $T \in S'$ non-empty, such that for all $t \in T$, $M,t \models \xi$. Then, we will show that $M',T \models \xi$. This implies that $M',S/\G{=1}\psi \models \G{=1}\psi$, which is what we actually want to prove. We apply induction over $\xi$. Note that the non-emptyness requirement of $T$ is important for the proof but holds for all $T \in S'$ in this fragment due to Lemma~\ref{lem:G1(Fq,G1)-subformulae}.

    \begin{enumerate}[label=\Roman*]
        \item $\xi=a$ or $\xi=\neg a$. Assume $\xi = a$. Then, for all states $t \in T$, $a \in L(t)$, and therefore $a \in L'(T)$. This yields $M',T \models a$. The case $\xi = \neg a$ is analagous.

        \item $\xi=\zeta \land \vartheta$. For all $t \in T$, $M,t \models \zeta$ and $M,t \models \vartheta$. By the induction hypothesis follows that $M',T \models \zeta$ and $M',T \models \vartheta$. Therefore, $M',T \models \zeta \land \vartheta$.

        \item $\xi=\F{\rhd r}{\zeta}$. According to the definition of $S'$, there is a $S/\zeta \in S'$. Due to Lemma~\ref{lem:G1(Fq,G1)-subformulae}, there is a state $t \in post^*(s)$, such that $M,t \models \zeta$. By the induction hypothesis, $M',S/\zeta \models \zeta$. Since $M'$ is a BSCC by definition, $S/\zeta$ is reached from $T$ almost surely. Therefore, $M,T \models \F{=1}\zeta$.

        \item $\xi=\G{=1}{\zeta}$. By definition, for all $T' \in S'$, $T' \subseteq post^*(s)$. By construction, $M,s \models \G{=1}\zeta$ and therefore, for every $t \in T'$, $M,t \models \zeta$. The induction hypothesis yields that $M',T' \models \zeta$. Therefore, $M',T' \models \G{=1}\zeta$ and in particular $M',T \models \G{=1}\zeta$.\QED
    \end{enumerate}
\end{proof}

\subsection{Satisfiability for $\F q,\G1$}%\label{ss:nodisj}
We extend the notion of subformulae to contain the temporal operators with all possible probabilities: We define the set $sub^*(\phi)$, such that $sub(\phi) \subseteq sub^*(\phi)$ and $\F{\rhd q}\psi \in sub^*(\phi)$ (or $\G{\rhd q}\psi \in sub^*(\phi)$) implies $\F{\rhd' q'}\psi \in sub^*(\phi)$ (or $\G{\rhd' q'}\psi \in sub^*(\phi)$) for all $\rhd' \in \{>,\geq\}, q' \in [0,1] \cap \mathbb{Q}$. Then, for two formulae $\phi, \psi$, let $\psi \prec^* \phi$ iff $\psi \in sub^*(\phi)$ and $\psi \not= \phi$.  Moreover, for a set of formulae $\Phi$, we denote

\begin{align*}
top(\Phi) := \{
    \phi \in \Phi \mid \phi = \F{\rhd q}{\xi} \text{ and for all } \psi \in \Phi. (
            &\phi \not\prec^* \psi \text{ or } \\
            &\psi = \phi \land \xi \text{ or } \\
            &\psi = \phi \lor \xi
        )
\}
\end{align*}

Intuitively, the set $top$ denotes the set of temporal top formulae with \F{} as top-most operator. Let CONSTRUCT\_G\_MODELS($s$) be the \G{}-model construction procedure as presented in the proof of Theorem~\ref{thm:G1(Fq,G1)-size}. Moreover, let PRUNE\_BRANCHES($s$) be a procedure that reduces the number of outgoing transitions from $s$ according to \cite{LICS}. Consider Algorithm~\ref{alg:Fq,G1-normal}.

\begin{algorithm}
    \caption{Normalization for the $\F q, \G 1$ fragment}
    \label{alg:Fq,G1-normal}
    \begin{algorithmic}
        \State REDUCE($\initstate, \emptyset$)

        \Function{REDUCE}{$s, F$}
            \State $I := top(sat(s)) \setminus F$

            \If{$I = \emptyset$}
                \If{$sat(s) = \emptyset$}
                    \State $P'(s,s) := 1$
                \Else
                    \State CONSTRUCT\_G\_MODELS($s$)
                \EndIf
            \EndIf

            \State
            \begin{align*}
                sel(s) := \{t \in post^*(s) \mid \exists \F{\rhd q} \psi \in I.\ (M,t \models \psi \land \forall t' \in pre^*(t) \cap post^*(s).\ M,t' \not\models \psi) \}
            \end{align*}

            \State $p_s := \sum_{t \in sel(s)}{P(s,t)}$

            \State $P'(s,t) :=
                \begin{cases}
                    P^*(s,t)/p_s &\text{ if } t \in sel(s) \\
                    0 &\text{ otherwise}
                \end{cases}$

            \State PRUNE\_BRANCHES($s$)

            \For{$t \in post_{M'}(s)$}
                \State $F' := F \cup \{\F{\rhd' q'}{\psi} \mid \F{\rhd q}{\psi} \in top(sat(s)) \land M,t \models \psi\}$
                \State REDUCE($t, F'$)
            \EndFor
        \EndFunction
    \end{algorithmic}
\end{algorithm}
\todo{Define $sat(s)$.}

Algorithm \ref{alg:Fq,G1-normal} is a reduction function for models. Let $M' := \mathrm{REDUCE}(s_0)$. For an MC, $\sum_{t \in post(s)}{P(s,t)} = 1$ must hold for all $s \in S$. Assuming that PRUNE\_BRANCHES($s$) preserves this property, we can easily verify that $M'$ is an MC, as well.

\[
    \sum_{t \in post_{M'}(s)}{P'(s,t)} = \sum_{t \in sel(s)}{P^*(s,t)/p_s} = 1/p_s \cdot p_s = 1
\]

\begin{lemma}
    \label{lem:cons-prob}
    Let $M$ be a tree MC, $M' := \mathrm{REDUCE}(\initstate)$, $s \in S'$, and $t \in post^*_{M'}(s)$. Then, $P'^*(s,t) \geq P^*(s,t)$.
\end{lemma}

\begin{proof}
    Let $\rho$ be the path leading from $s$ to $t$ in $M'$.  Note that we can assume that $\rho$ is unique, since $M$ is a tree and the reduction preserve this property. We will show the claim by induction over $\rho$.

    \begin{enumerate}[label=\Roman*]
        \item $\rho = st$.
            Then, $P'^*(s,t) = P'(s,t) \overset{def.}{=} P^*(s,t)/p_s \geq P^*(s,t)$. The last inequality follows from the fact that $p_s \in [0,1]$.

        \item $\rho = s\rho't$, where $len(\rho') > 0$.  Let $s' := \rho'[0]$.

            \begin{align*}
                P'^*(s,t) &= P'^*(s,s') \cdot P'^*(s',t) \\
                    &= P'(s,s') \cdot P'^*(s',t) \\
                    &\overset{def.}{=} P^*(s,s')/p_s \cdot P'^*(s',t) \\
                    &\geq P^*(s,s') \cdot P'^*(s',t) \\
                    &\overset{I.H.}{\geq} P^*(s,s') \cdot P^*(s',t) \\
                    &= P^*(s,t) \QED
            \end{align*}
    \end{enumerate}
\end{proof}

\begin{lemma}
    \label{lem:Fq,G1-normal-size}
    Let $M$ be a tree MC. Then, the height of $M' := \mathrm{REDUCE}(\initstate,\emptyset)$ is bounded by $|\phi|$.
\end{lemma}

\begin{proof}
    For a set of formulae $\Phi$, let $max(\Phi) := \{ \psi \in \Phi \mid \nexists \xi \in \Phi.\ \psi \prec^* \xi \}$. This is the set of maximal elements w.r.t $\prec^*$. Let $s \in S'$. We apply induction over $n_{s,F} := |sub(top(sat(s)) \setminus F)|$ and show that the height is bounded by $k_{s,F}:= |sub(max(sat(s))\setminus F)|$. Note that if $\phi \in sat(s)$, then $max(sat(s)) = \{\phi\}$. Moreover, we assume w.l.o.g. that $t \in post(s)$ implies $sub(sat(t)) \subseteq sub(sat(s))$.

    \begin{enumerate}[label=\Roman*]
        \item $n=0$. Then, $top(sat(s)) \setminus F = \emptyset$. This might be due to several reasons. Firstly, it might be the case that $sat(s)=\emptyset$. If this is the case, then a self loop is produced, which we consider to be a tree of size $0 \leq k_{s,F}$. If $sat(s)\not=\emptyset$, then we know at least that there are no formulae left in $top(sat(s))$ which are not in $F$. In that case, we construct the models for the \G{}-formulae. By Theorem~\ref{thm:G1(Fq,G1)-size} we know that those are bounded by the size of the \G{}-formula (we can assume that there is only one because $\G{=1}$ distributes over conjunctions---see Appendix~\ref{app:eqs} for the proof). Therefore, the result is again of height bounded by $k_{s,F}$.

        \item $n>0$. In that case, we directly connect $s$ with $sel(s)$, where $sel(s)$ enforces the satisfaction of the argument of at least one \F{}-formula in $I$. Let $t \in post_{M'}(s)$ and $M,t \models \psi$, where $\F{\rhd q}\psi \in I$. Then, $\F{\rhd' q'}\psi$ is added to $F'$. Hence, $F \subset F'$ and $n_{t,F'} < n_{s,F}$. By the induction hypothesis follows that the height from $t$ is bounded by $k_{t,F'} < k_{s,F}$. Therefore, the height from $s$ is bounded by $k_{s,F}$.
    \end{enumerate}
\QED
\end{proof}

\begin{lemma}
    \label{lem:Fq,G1-normal}
    Let $\phi$ be a $\F q,\G 1$ formula, and $M$ a model for $\phi$. Then, $M' := \mathrm{REDUCE}(\initstate,\emptyset)$ is also a model for $\phi$.
\end{lemma}

\begin{proof}
    Let $\psi \in sub(\phi)$, and $s \in S'$. Moreover, let $F$ denote the set of non-\G{}-nested \F{}-formulae whose argument has been satisfied on the path from $s_0$ to $s$. We apply induction over $\psi$ to show that if $M,s \models \psi$, then
    \begin{enumerate}
        \item $\psi = a$ implies $M',s \models \psi$,
        \item $\psi = \xi \land \zeta$ implies that the claim holds for $\xi$ and $\zeta$,
        \item $\psi = \F{\rhd q}\xi$ and $\psi \in I_s$ implies $M',s \models \psi$,
        \item $\psi$ is \G{}-nested implies $M',s \models \psi$,
        \item $\psi = \G{=1}\xi$ implies $M',s \models \psi$,
    \end{enumerate}
    
    where $I_s:= top(sat(s)) \setminus F$. From the above conditions follows that $M',s_0 \models \phi$.

    \begin{enumerate}[label=\Roman*]
        \item $\psi = a$. Since $L'(s) = L(s)$, $M',s \models a$ iff $M,s \models a$.

        \item $\psi = \xi \land \zeta$. We assume that $M,s \models \psi$ and therefore $M,s \models \xi$ and $M,s \models \zeta$. By the induction hypothesis, the claim holds for $\xi$ and for $\zeta$.

        \item $\psi = \F{\rhd r}\xi$. If $\psi \in I_s$, then the construction of $M'$ enforces that the first states to satisfy $\xi$ are included on every branch. By Lemma~\ref{lem:cons-prob}, those are reached with at least the probability they are reached with in $M$. Hence, $M,s \models \psi$ implies $M',s \models \psi$. If $\psi$ is \G{}-nested, then the claim follows from Corollary~\ref{cor:G1(Fq,G1)-bsccs} and Lemma~\ref{lem:bsccs}.

        \item $\psi = \G{=1}\xi$. Then, for all $t \in post^*_M(s)$, $M,t \models \xi$. By the induction hypothesis, $M',t \models \xi$. Since $post^*_{M'}(s) \subseteq post^*_M(s)$, $M',s \models \psi$.\QED
    \end{enumerate}
\end{proof}

\begin{reftheorem}{thm:Fq,G1-size}
    A satisfiable, $\F q,\G 1$-formula $\phi$ has a model of size $f(|\phi|)$, for some computable function $f$.
\end{reftheorem}

\begin{proof}
    This follows immediately from Lemmas~\ref{lem:Fq,G1-normal-size} and \ref{lem:Fq,G1-normal}.
\QED\end{proof}

\subsection{Satisfiability for $\G1(\F q,\G1,\vee)$}

\begin{reftheorem}{thm:G1(Fq,G1,v)-size}
    Let $\phi := \G{=1}\psi$, and $M$ a model for $\phi$. Then, there is a model for $\phi$ of linear size in $|\phi|$.
\end{reftheorem}

\begin{proof}
    Let $G := \{ \xi \in sub(\phi) \mid \xi = \G{=1}{\xi} \text{ for some } \xi \}$, and $s \in S$, such that $|G \setminus sat(s)| = min_{t \in S}(|G \setminus sat(t)|)$. As in the proof of Theorem~\ref{thm:G1(Fq,G1)-size}, let $S/\xi := \{ t \in post^*(s) \mid M,t \models \xi \}$, $S' := \{ S/\xi \mid \xi \in sub(\phi) \land S/\xi \not= \emptyset \}$, $L'(T) := \bigcap_{t \in T}(L(t))$, and $P'$ generating a BSCC from $S'$. Note that in the proof of Theorem~\ref{thm:G1(Fq,G1)-size} we did not include the non-emptyness requirement on $S/\xi$. This was due to the fact that every such set was non-empty due to Lemma~\ref{lem:G1(Fq,G1)-subformulae}. Here, this is not the case because of the disjunctions. However, if for some subformula, there is no state in $M$ that models it, then we can safely ignore it.

    Now, we again prove that for $\xi \in sub(\phi)$, and $T \in S'$, where for all $t \in T$, $M,t \models \xi$, it holds that $M',T \models \xi$. The proof works by induction over $\xi$. All steps are essentially the same as in the proof of Theorem~\ref{thm:G1(Fq,G1)-size}. The only additional case is $\xi = \zeta \lor \vartheta$. However, the proof in this case is analagous to the case $\xi = \zeta \lor \vartheta$.
\QED
\end{proof}

\subsection{Satisfiability for $\F{q/1},\G1,\vee$}%\label{ss:noqf}
\begin{align*}
    &\Phi ::= a \mid
             \neg a \mid
             \Phi \land \Phi \mid
             \Phi \lor \Phi \mid
             \PF{\Phi}{\rhd r} \mid
             \PG{\Psi}{=1} \\
    &\Psi ::= a \mid
             \neg a \mid
             \Psi \land \Psi \mid
             \Psi \lor \Psi \mid
             \PF{\Psi}{=1} \mid
             \PG{\Psi}{=1}
\end{align*}

\paragraph*{Construction of reduced models}
Let $\phi$ be a $\F {q/1},\G 1,\vee$ formula, and $M$ be a model for $\phi$. Further, let $M' := \mathrm{REDUCE}(\initstate, \emptyset)$. We have already shown that $M'$ is limited in size. Moreover, in the conjunctive setting we have also shown that $M'$ is still a model for $\phi$. In this setting, this might not be the case anymore. The reason is that we cannot guarantee that all formulae that are nested in \G{}s are necessarily satisfied in all BSCCs. Therefore, we have to take care of those, as well.

For this, we will use $M$, in order to expand $M'$, such that $M',\initstate \models \phi$. In a sense, we will learn from the states that we omitted in $M'$. Let $s \in S$, and $t \in post_{M'}(s)$. We denote $\rho^M_{st}$ for the unique path leading from $s$ to $t$ in $M$. Slightly abusing notation, we denote

$$
    \rho^M_{st}[\psi] := \begin{cases}
        \rho^M_{st}[i] &\text{ if } M,\rho^M_{st}[i] \models \psi \text{ and } \nexists j>i. M,\rho^M_{st}[j] \models \psi\\
        \mathtt{undefined} &\text{ otherwise}
    \end{cases}
$$

for the last state on $\rho^M_{st}$ that satisfies $\psi$. Using this, we can determine which states we need to add to $M'$ in order to obtain a model, namely:

$$
    T(\rho^M_{st}) := \{\rho^M_{st}[\psi] \mid \exists \G{=1}{\xi} \in sub(\phi).\ \F{=1}{\psi} \in sub(\xi)\}
$$

Now, we can define our model $\tilde{M}$, with

$$
\tilde{S} := S' \cup \bigcup_{\substack{s \in S' \\ t \in post_{M'}(s)}}{T(\rho^M_{st})}
$$

$$
\tilde{L} := L|_{\tilde{S}}
$$

$$
    \tilde{P}(s,t) := \begin{cases}
        P'(s,t) &\text{if } s,t \in S' \text{ and } T(\rho^M_{st})=\emptyset\\
        P'(s,t') &\text{if } t \in T(\rho^M_{st'}) \text{ and }
        T(\rho^M_{st'}) \cap pre^*_M(t) = \emptyset \\
        1 &\text{if } s,t \in T(\rho^M_{s't'}), s \in pre^*_M(t)
        \text{ and }\\
          &T(\rho^M_{s't'}) \cap post^*_M(s) \cap pre^*_M(t) = \emptyset \\
        1 &\text{if } s \in T(\rho^M_{s't}) \text{ and }
            T(\rho^M_{s't}) \cap post^*_M(s) = \emptyset \\
        0 &\text{otherwise }
    \end{cases}
$$

\paragraph*{Proof of correctness}
$\tilde{M}$ extends $M'$ by simple chains of states which satisfy the arguments of \G{}-nested \F{}-formulae. The construction of $\tilde{P}$ preserves the probabilities in a certain sense, as the following lemma states.

\begin{lemma}
    \label{lem:Fq,G1,v-cons-prob}
    Let $s,t \in S'$. Then, $\tilde{P}^*(s,t) = P'^*(s,t)$.
\end{lemma}

\begin{proof}
    We apply induction over $\rho^{M'}_{st}$.

    \begin{enumerate}[label=\Roman*]
        \item $\rho^{M'}_{st} = st$. We have to show that $\tilde{P}^*(s,t) = P'(s,t)$. If there are no states in between $s$ and $t$ in $\tilde{M}$, then by definition $\tilde{P}(s,t) = P'(s,t)$. Otherwise, let $s' \in post_{\tilde{M}}(s) \cap pre^*_{\tilde{M}}(t)$ be the first state to be inserted between $s$ and $t$. Then, $\tilde{P}(s,s') = P'(s,t)$. For all states $t' \in post^*_{\tilde{M}}(s') \cap pre^*_{\tilde{M}}(t)$, the construction guarantees that $\tilde{P}^*(s',t') = 1$ and then also that $\tilde{P}^*(s',t) = 1$. Hence, $\tilde{P}^*(s,t) = P'(s,t)$.

        \item $\rho^{M'}_{st} = s \rho t$, with $|\rho| > 0$. Let $s' := \rho[0]$. By the induction hypothesis, $\tilde{P}^*(s',t) = P'^*(s',t)$. Moreover, applying the same argument as in the base case we obtain that $\tilde{P}^*(s,s') = P'(s,s')$. Therefore, $\tilde{P}^*(s,t) = \tilde{P}^*(s,s') \cdot \tilde{P}^*(s',t) = P'(s,s') \cdot P'(s',t) = P'^*(s,t)$. \QED
    \end{enumerate}
\end{proof}

\begin{reftheorem}{thm:Fq1,G1,v-size}
    For a satisfiable $\F {q/1},\G 1,\vee$-formula $\phi$, there is a model of size $f(|\phi|)$, where $f$ is a computable function.
\end{reftheorem}

\begin{proof}
    We claim that $\tilde{M},\initstate \models \phi$. Let $\psi \in sub(\phi)$, and $s \in \tilde{S}$, such that $M,s \models \psi$. We show that
    \begin{itemize}
        \item if $\psi$ is \G{}-nested, then $\tilde{M},s \models \psi$, and
        \item otherwise if additionally $s \in S'$, then $\tilde{M},s \models \psi$.
    \end{itemize}
    We apply induction over $\psi$.

    \begin{enumerate}[label=\Roman*]
        \item $\psi = a$. Since $\tilde{L}(s) = L(s)$, $\tilde{M},s \models \psi$.

        \item $\psi = \xi \land \zeta$. Then, $M,s \models \xi$ and $M,s \models \zeta$. By the induction hypothesis, $\tilde{M},s \models \xi$ and $\tilde{M},s \models \zeta$. Hence, $\tilde{M},s \models \xi \land \zeta$. This argument works irrespective of whether $\psi$ is \G{}-nested or not.

        \item $\psi = \xi \lor \zeta$. Analagous to the previous case.

        \item $\psi = \F{\rhd q}\xi$. Here, we have to distinguish whether $\psi$ is \G{}-nested or not. If it is not, then $\tilde{M},s \models \F{\rhd q}\xi$ follows from Lemma~\ref{lem:Fq,G1,v-cons-prob} and a similar argument as in the proof of Theorem~\ref{thm:Fq,G1-size}. Otherwise, the constraints on \G{}-nested \F{}s enforce that $\psi = \F{=1}\xi$. Hence, we know that we can find $\xi$ on almost all paths from $s$. The construction of $\tilde{M}$ guarantees that we include such states on every path. Hence, $\tilde{M},s \models \F{=1}\xi$.

        \item $\psi = \G{=1}\xi$. First, observe that the construction of $\tilde{M}$ preserves the order of states; i.e. for $s,t \in \tilde{S}$ with $t \in post^*_[\tilde{M}](s)$ it is also the case that $t \in post^*_M(s)$. We know that for all $t \in post^*_M(s)$, $M,t \models \xi$. Then, by the induction hypothesis it also follows that $\tilde{M},t \models \xi$. Therefore, $\tilde{M},s \models \G{=1}\xi$.
    \end{enumerate}
    
    Now, we still have to argue that the size of $\tilde{M}$ is bounded. We know that the size of $M'$ is bounded. In $\tilde{M}$ we inserted each \G{}-nested \F{}-formula at most once in between every two states of $M'$. Hence, the height is again bounded.
\QED\end{proof}

\section{Proof of Theorem~\ref{thm:G1(Fq,G1)-normal}}
\label{app:eqs}
In this section, we will use a slightly different notion of models. It is equivalent to the one that we have defined earlier. Let $\mathcal{F}$ be the set of PCTL formulae. Then, we define a labelled MC as usual, except that $L : S \rightarrow 2^{\mathcal{F}}$. Then, we define the PCTL semantics as follows:

\begin{definition}[PCTL Semantics]
    Let $M$ be a labelled MC and $s \in S$. We define the modeling relation $\models$ as follows

    \begin{enumerate}[label=(MS\arabic*),align=left]
        \item \label{def:model.a} $M, s \models a$ iff $a \in L(s)$.
        \item \label{def:model.neg} $M, s \models \neg \phi$ iff $M,s
            \not \models \phi$.
        \item \label{def:model.and} $M, s \models \phi \land \psi$ iff
            $M,s \models \phi$ and $M,s \models \psi$.
        \item \label{def:model.or} $M,s \models \phi \lor \psi$ iff $M,s \models
            \phi$ or $M,s \models \psi$.
        \item \label{def:model.F} $M, s \models \F{\rhd r}{\phi}$ iff
            $\pr_{M(s)}(\{ \pi \mid \exists i \in \mathbb{N}_0: M,\pi[i] \models
            \phi \}) \rhd r$
        \item \label{def:model.G} $M, s \models \G{\rhd r}{\phi}$ iff
            $\pr_{M(s)}(\{ \pi \mid \forall i \in \mathbb{N}_0: M,\pi[i] \models
            \phi \}) \rhd r$
    \end{enumerate}

    $M$ is a model, if for all formulae $\xi$, and all states $t \in S$, $\xi \in L(t)$ iff $M,t \models \xi$.
\end{definition}

\begin{lemma}
    \label{lem:FG-eqs}
    \begin{align}
        \G{=1}\G{=1} \psi & \equiv 
            \G{=1} \psi \label{eq:GG} \\
        \G{=1}{\F{\rhd r}{\psi}} & \equiv_{fin}
            \G{=1}{\F{=1}{\psi}} \label{eq:Gf} \\
        \F{=1}{\F{\rhd r}{\psi}} & \equiv 
            \F{\rhd r}{\psi} \label{eq:Ff} \\
        \F{=1}{\G{=1}{\F{=1}{\psi}}} & \equiv
            \G{=1}{\F{=1}{\psi}} \label{eq:FGF} \\
        \G{=1}{\F{=1}{\G{=1}{\psi}}} & \equiv
            \F{=1}{\G{=1}{\psi}} \label{eq:GFG}
    \end{align}
\end{lemma}

\begin{proof}
    Let $M$ be a model. First, observe that $\G{=1}{\psi} \in L(s_0)$
    implies that $\psi \in L(s_0)$, whereas $\psi \in L(s_0)$ implies
    that $\F{=1}{\psi} \in L(s_0)$.

    \paragraph*{Equality \eqref{eq:GG}}
    Assume that $\G{=1}{\G{=1}{\psi}} \in L(s_0)$. Then, due to the above
    observation, $\G{=1}{\psi} \in L(s_0)$. Now assume that $\G{=1}{\psi} \in
    L(s_0)$. Then, $\G{=1}{\psi} \in L(s)$, for all $s \in post^*(s_0)$.
    Therefore, $\pr(\{\pi \mid \pi \models \GG{\G{=1}{\psi}}\}) = 1$, which
    implies $\G{=1}{\G{=1}{\psi}} \in L(s_0)$.

    \paragraph*{Equality \eqref{eq:Gf}}
    Assume that $\G{=1}\F{\rhd r}\psi \in L(s_0)$. Then, for all BSCCs
    $T$, and all states $t \in T$, $\G{=1}\F{\rhd r}\psi \in L(t)$,
    and therefore $\F{\rhd r}\psi \in L(t)$. Hence, there must be a
    state $t' \in T$, with $\psi \in L(t')$. Since $T$ is a BSCC, $t'$
    is reached almost surely from every state in $T$. Therefore,
    $\G{=1}\F{=1}\psi \in L(t)$. As we are considering finite models
    only, every run ends up in a BSCC almost surely and thus
    $\G{=1}\F{=1}\psi \in L(s_0)$.  The reverse implication is
    obvious.

    \paragraph*{Equality \eqref{eq:Ff}}
    Assume $\F{=1}{\F{\rhd r}{\psi}} \in L(s_0)$. Then,
    there is a set $T \subseteq post^*(s_0)$, such that for all $t \in
    T$, $\F{\rhd r}{\psi} \in L(t)$ and $\pr(\{ \pi \mid \exists i. \pi[i] \in T
    \}) = 1$. We can compute
    the probability to reach $\psi$ as follows

    $$
    \pr(\{ \pi \mid \pi \models \FF{\psi} \}) \rhd
        \sum_{t \in T}{P^*(s_0,t)} \cdot r = r
    $$

    This means that $\F{\rhd r}{\psi} \in L(s_0)$. The converse
    implication follows immediately from the fact that $\xi \in
    L(s_0)$ implies $\F{=1}{\xi} \in L(s_0)$.

    \paragraph*{Equality \eqref{eq:FGF}}
    Assume $\F{=1}{\G{=1}{\F{=1}{\psi}}} \in L(s_0)$.
    Then, there is a set $T \subseteq post^*(s)$, where for all $t \in
    T$, $\G{=1}{\F{=1}{\psi}} \in L(t)$ and $\pr(\{ \pi \in
    Cyl(s_0) \mid \exists i. \pi[i] \in T \}) = 1$. Thus, for all $t
    \in T$, $\F{=1}{\psi} \in L(t)$. Therefore, $\F{=1}{\F{=1}{\psi}}
    \in L(s_0)$.  Equality \eqref{eq:Ff} yields $\F{=1}{\psi} \in
    L(s_0)$. This argument can be applied to all states in
    $post^*(s_0) \cap pre^*(T)$. Therefore, for all $s \in
    post^*(s_0)$, $\F{=1}{\psi} \in L(s)$ and then
    $\G{=1}{\F{=1}{\psi}} \in L(s_0)$.  The converse implication is
    again due to $\xi \Rightarrow \F{=1}{\xi}$.

    \paragraph*{Equality \eqref{eq:GFG}}
    Assume $\F{=1}{\G{=1}{\psi}} \in L(s_0)$. Then, for
    every state $s \in post^*(s_0)$, either
    $\F{=1}{\G{=1}{\psi}} \in L(s)$ or $\G{=1}{\psi}
    \in L(s)$. In the latter case, however, it also holds that
    $\F{=1}{\G{=1}{\psi}} \in L(s_0)$. Therefore, for every
    state $s \in post^*(s_0)$, $\F{=1}{\G{=1}{\psi}} \in
    L(s)$, and thus $\G{=1}{\F{=1}{\G{=1}{\psi}}} \in
    L(s_0)$.  The converse implication follows from $\G{=1}{\psi}
    \Rightarrow \psi$.\QED

\end{proof}

\begin{lemma}[Distributivity]
    \label{lem:distr}
    \begin{align}
        \F{=1}{(\bigwedge_i {\F{\rhd r_i}{\psi_i}})} & \equiv
            \bigwedge_i {\F{\rhd r_i}{\psi_i}} \label{eq:FCf} \\
        \G{=1}{(\bigwedge_i {\psi_i})} & \equiv
            \bigwedge_i {\G{=1}{\psi_i}} \label{eq:GC} \\
        \F{=1}{\G{=1}{(\bigwedge_i {\psi_i})}} & \equiv
            \bigwedge_i {\F{=1}{\G{=1}{\psi_i}}} \label{eq:FGC} \\
        \G{=1}{\F{=1}(\psi \land \F{\rhd r}{\xi})} &
            \equiv_{fin} \G{=1}(\F{=1}{\psi} \land
            \F{=1}{\xi}) \label{eq:GFCf} \\
        \G{=1}\F{=1}(\psi \land \G{=1}{\xi}) &
            \equiv \G{=1}(\F{=1}{\psi} \land
            \F{=1}{\G{=1}{\xi}}) \label{eq:GFCG}
    \end{align}
    
\end{lemma}

\begin{proof}
    Let $M$ be a model. In general, 
    
    \begin{equation}
        \label{imp:FC}
        \F{=1}{\bigwedge_i{\psi_i}} \Rightarrow
        \bigwedge_i{\F{=1}{\psi_i}}
    \end{equation}

    \paragraph*{Equality \eqref{eq:FCf}}
    From the implication \eqref{imp:FC} follows

    $$
    \F{=1}{\bigwedge_i{\F{\rhd r_i}{\psi_i}}} \Rightarrow
    \bigwedge_i{\F{=1}{\F{\rhd r_i}{\psi_i}}}
    \overset{\eqref{eq:Ff}}{\equiv} \bigwedge_i{\F{\rhd r_i}{\psi_i}}
    $$
    
    The converse implication is clear.

    \paragraph*{Equality \eqref{eq:GC}}
    Assume $\bigwedge_i{\G{=1}{\psi_i}} \in L(s_0)$. Then, for all
    $s \in post^*(s_0)$, and all $i$, $\psi_i \in L(s)$. This implies
    that $\bigwedge_i{\psi_i} \in L(s)$ and therefore
    $\G{=1}{\bigwedge_i{\psi_i}} \in L(s_0)$.

    If, on the other hand, $\G{=1}{\bigwedge_i{\psi_i}} \in
    L(s_0)$, then for every $s \in post^*(s_0)$, $\bigwedge_i{\psi_i}
    \in L(s)$, and hence, for all $i$, $\psi_i \in L(s)$. Thus
    $\bigwedge_i{\G{=1}{\psi_i}} \in L(s_0)$.

    \paragraph*{Equality \eqref{eq:FGC}}
    From equality \eqref{eq:GC} follows that
    $\F{=1}{\G{=1}{\bigwedge_i{\psi_i}}} \equiv
    \F{=1}{\bigwedge_i{\G{=1}{\psi_i}}}$, and from the
    implication \eqref{imp:FC},
    $\F{=1}{\bigwedge_i{\G{=1}{\psi_i}}} \Rightarrow
    \bigwedge_i{\F{=1}{\G{=1}{\psi_i}}}$.

    Now assume that $\bigwedge_i{\F{=1}{\G{=1}{\psi_i}}} \in
    L(s_0)$. Then, for all $i$, there is a set $T_i \subseteq
    post^*(s_0)$, where for all $t \in T_i$, $\G{=1}{\psi_i} \in L(t)$
    and $\pr(\{ \pi \mid \exists j. \pi[j] \in T_i \}) =
    1$. Let $T := \bigcap_i{T_i}$. Then, for all $t \in T$, and all
    $i$, $\G{=1}{\psi_i} \in L(t)$. What is left to show is that
    $\pr(\{ \pi \mid \exists j. \pi[j] \in T \}) = 1$.
    Since for all $s \in post^*(s_0)$, either $\F{=1}{\G{=1}{\psi_i}}
    \in L(s)$ or $\G{=1}{\psi_i} \in L(s)$, the same holds for every
    $T' \subseteq pre^*(T)$, and in particular for every $T_i$. Thus,
    $T$ is reached almost surely.

    \paragraph*{Equality \eqref{eq:GFCf}}
    Assume $\G{=1}(\F{=1}{\psi} \land \F{=1}{\xi}) \in
    L(s_0)$. Then, for all states $s \in post^*(s_0)$,
    $\F{=1}{\psi} \in L(s)$ and $\F{=1}{\xi} \in L(s)$. Let $s'$ be
    such that $\psi \in L(s)$. Then, $\F{=1}{\xi} \in L(s')$ must
    also hold, and therefore $\psi \land \F{=1}{\xi} \in L(s')$.
    Since this is true for all states that satisfy $\psi$, and those
    are reached almost surely from every state, $\G{=1}(\F{=1}(\psi
    \land \F{=1}{\xi})) \in L(s_0)$.

    For the converse implication, we can apply our proven equalities
    to obtain

    \begin{align*}
        \G{=1}(\F{=1}(\psi \land \F{\rhd r}{\xi}))
            &\overset{\eqref{imp:FC}}{\Rightarrow}
        \G{=1}(\F{=1}{\psi} \land \F{=1}{\F{\rhd r}{
            \xi}})\\
            &\overset{\eqref{eq:Ff}}{\equiv}
        \G{=1}(\F{=1}{\psi} \land \F{\rhd r}{\xi})\\
            &\overset{\eqref{eq:GC}}{\equiv}
        \G{=1}(\F{=1}{\psi}) \land \G{=1}(\F{\rhd r}{
            \xi})\\
            &\overset{\eqref{eq:Gf}}{\equiv_{fin}}
        \G{=1}(\F{=1}{\psi}) \land \G{=1}(\F{=1}{
            \xi})\\
            &\overset{\eqref{eq:GC}}{\equiv}
        \G{=1}(\F{=1}{\psi} \land \F{=1}{\xi})
    \end{align*}

    \paragraph*{Equality \eqref{eq:GFCG}}
    Assume $\G{=1}(\F{=1}{\psi} \land
    \F{=1}{\G{=1}{\xi}}) \in L(s_0)$. Then there is a
    set $T \subseteq post^*(s_0)$, such that for all $t \in T$,
    $\G{=1}{\xi} \in L(t)$, and $\pr(\{ \pi \mid
    \exists i. \pi[i] \in T \}) = 1$. Since all successors of $s_0$
    satisfy $\F{=1}{\psi}$, in particular this must be true for
    all $t \in T$. Hence, there must be $T' \subseteq T$, with all
    states satisfying $\psi$ and $\pr(\{ \pi  \mid
    \exists i. \pi[i] \in T' \}) = 1$. Therefore, $\F{=1}(\psi \land
    \G{=1}{\xi}) \in L(s_0)$. This argument can be applied to
    every successor of $s_0$, and we therefore conclude that
    $\G{=1}{\F{=1}(\psi \land \G{=1}{\xi})} \in
    L(s_0)$.  The converse implication follows immediately from the
    implication \eqref{imp:FC}.\QED
\end{proof}

Now we can provide an alternative proof for Theorem~\ref{thm:G1(Fq,G1)-normal}. Recall the theorem statement.

\begin{reftheorem}{thm:G1(Fq,G1)-normal}
    Let $\phi$ be a conjunctive \FGfrag{q}{1}-formula. Then, the following equality holds

    $$
    \G{=1}(\phi) \equiv_{fin}
        \G{=1}(\bigwedge_{l \in A}{l} \land
        \F{=1}{\G{=1}(\bigwedge_{l \in B}{l})} \land
        \bigwedge_{i \in I}\F{=1}({\bigwedge_{l \in C_i}{l}}))
    $$

    For appropriate $A, B, C_i \subset \literals$.
\end{reftheorem}

\begin{proof}
    We apply induction over $\phi$.

    \paragraph*{Case $\phi = l$}
    Then $A := \{l\}$ and the claim holds.

    \paragraph*{Case $\phi \equiv \psi \land \xi$}
    
    \begin{align*}
        &\G{=1}(\psi \land \xi)\\
        &\overset{\eqref{eq:GC}}{\equiv}
            \G{=1}{\psi} \land \G{=1}{\xi}\\
        &\overset{I.H}{\equiv_{fin}}
            \G{=1} (\bigwedge_{l \in A_{\psi}}{l} \land
            \F{=1} \G{=1} (\bigwedge_{l \in B_{\psi}}{l})
            \land \bigwedge_{i \in I_{\psi}}{\F{=1} (\bigwedge_{l
            \in C_{\psi,i}}{l}}))\\
        &\land
            \G{=1} (\bigwedge_{l \in A_{\xi}}{l} \land \F{=1} \G{=1}
            (\bigwedge_{l \in B_{\xi}}{l}) \land \bigwedge_{i \in
            I_{\xi}}{\F{=1} (\bigwedge_{l \in C_{\xi,i}}{l}}))\\
        &\overset{\eqref{eq:GC}}{\equiv}
            \G{=1} (\bigwedge_{l \in A_{\psi} \cup A_{\xi}}{l} \land
            \F{=1} \G{=1} (\bigwedge_{l \in B_{\psi}}{l}) \land \F{=1}
            \G{=1} (\bigwedge_{l \in
            B_\xi}{l})\\
        &\land
            \bigwedge_{i \in I_{\psi}}{\F{=1} (\bigwedge_{l \in
            C_{\psi,i}}{l})} \land \bigwedge_{i \in I_{\xi}}{\F{=1}
            (\bigwedge_{l \in C_{\xi,i}}{l})})\\
        &\overset{\eqref{eq:FGC}}{\equiv}
            \G{=1} (\bigwedge_{l \in A_{\psi} \cup A_{\xi}}{l} \land
            \F{=1} \G{=1} (\bigwedge_{l \in B_{\psi} \cup
            B_{\xi}}{l})\\
        &\land
            \bigwedge_{i \in I_{\psi}}{\F{=1} (\bigwedge_{l \in
            C_{\psi,i}}{l}}) \land \bigwedge_{i \in
            I_{\xi}}{\F{=1} (\bigwedge_{l \in C_{\xi,i}}{l}}))
    \end{align*}

    \paragraph*{Case $\phi = \G{=1}{\psi}$}
    Then from equality \eqref{eq:GG} it follows that

    $$\G{=1}{\G{=1}{\psi}} \equiv \G{=1}{\psi}$$
    
    and thus the claim holds by induction hypothesis.

    \paragraph*{Case $\phi = \F{\rhd r}{\psi}$}
    From equality \eqref{eq:Gf} it follows
    
    $$\G{=1}{\F{\rhd r}{\psi}} \equiv_{fin} \G{=1}{\F{=1}{\psi}}$$
    
    This case will be covered next.

    \paragraph*{Case $\phi = \F{=1}{\psi}$}
    Now we will show that

    $$
    \G{=1}{\F{=1}{\psi}} \equiv
    \G{=1}(\F{=1}\G{=1}(\bigwedge_{l \in B}{l}) \land \bigwedge_{i
        \in I}{\F{=1}(\bigwedge_{l \in C_i}{l})})
    $$

    For this, we will consider several subcases, which results in
    another induction. In order to distinguish between the hypotheses,
    we will refer to the inductions as the \emph{inner} and
    \emph{outer} induction, respectively.

    \textit{Subcase $\psi = l$}
    Setting $C := \{l\}$ yields the claim.

    \textit{Subcase $\psi = \F{\rhd r}{\xi}$}
    Applying Lemma \ref{lem:FG-eqs}, we get

    \begin{align*}
        \G{=1}{\F{=1}{\F{\rhd r}{\xi}}}
        &\overset{\eqref{eq:Ff}}{\equiv}
        \G{=1}{\F{\rhd r}{\xi}} \\
        &\overset{\eqref{eq:Gf}}{\equiv_{fin}}
        \G{=1}{\F{=1}{\xi}}
    \end{align*}

    and the claim holds by inner induction hypothesis.

    \textit{Subcase $\psi = \G{=1}{\xi}$}
    Applying the outer induction hypothesis and Lemmas
    \ref{lem:FG-eqs} and \ref{lem:distr} yields

    \begin{align*}
        &\G{=1}{\F{=1}{\G{=1}{\xi}}}\\
        &\overset{o.I.H}{\equiv_{fin}}
            \G{=1} \F{=1} \G{=1} (\bigwedge_{l \in A}{l}
            \land \F{=1} \G{=1} (\bigwedge_{l \in B}{l})
            \land \bigwedge_{i \in I}{\F{=1} (\bigwedge_{l \in
            C_i}{l}})) \\
        &\overset{\eqref{eq:FGC}}{\equiv}
            \G{=1} (\F{=1} \G{=1} (\bigwedge_{l \in
            A}{l})
            \land \F{=1} \G{=1} \F{=1} \G{=1} (\bigwedge_{l \in B}{l})
            \\
        &\land
            \bigwedge_{i \in I}{\F{=1} \G{=1} \F{=1}
            (\bigwedge_{l \in C_i}{l})}) \\
        &\overset{\substack{\eqref{eq:FGF} \\
        \eqref{eq:GFG}}}{\equiv}
            \G{=1} (\F{=1} \G{=1} (\bigwedge_{l \in
            A}{l})
            \land \F{=1} \G{=1} (\bigwedge_{l \in B}{l})
            \land \bigwedge_{i \in I} {\G{=1} \F{=1} (\bigwedge_{l \in
            C_i}{l})}) \\
        &\overset{\eqref{eq:FGC}}{\equiv}
            \G{=1} (\F{=1} \G{=1} (\bigwedge_{l \in A \cup B}{l})
            \land \bigwedge_{i \in I}{\G{=1} \F{=1} (\bigwedge_{l \in
            C_i}{l}})) \\
        &\overset{\eqref{eq:GC}}{\equiv}
            \G{=1} (\F{=1} \G{=1} (\bigwedge_{l \in A \cup B}{l}))
            \land
            \G{=1} \G{=1} (\bigwedge_{i \in I}{\F{=1} (\bigwedge_{l
            \in C_i}{l})}) \\
        &\overset{\eqref{eq:GG}}{\equiv}
            \G{=1} (\F{=1} \G{=1} (\bigwedge_{l \in A \cup B}{l}))
            \land
            \G{=1} (\bigwedge_{i \in I}{\F{=1} (\bigwedge_{l
            \in C_i}{l})}) \\
        &\overset{\eqref{eq:GC}}{\equiv}
            \G{=1} (\F{=1} \G{=1} (\bigwedge_{l \in A \cup B}{l})
            \land \bigwedge_{i \in I}{\F{=1} (\bigwedge_{l \in
            C_i}{l}})) \\
    \end{align*}

    \textit{$\psi \equiv \bigwedge_i \xi_i$}
    In this case we need to show that although $\F{=1}$ does not
    distribute over conjunctions in general, we still can get a
    formula of the desired form. We can split the conjunction into
    subformulae like this:

    \begin{equation}
        \bigwedge_i {\xi_i} \equiv \bigwedge_{l \in C}{l} \land
        \bigwedge_i {\F{\rhd r_i} \zeta_i} \land \bigwedge_i {\G{=1}
        \vartheta_i}
    \label{eq:C}
    \end{equation}

    Then we can apply the induction hypotheses and the lemmas
    \ref{lem:FG-eqs} and \ref{lem:distr} to obtain:

    \begingroup
    \allowdisplaybreaks
    \begin{align*}
        &\G{=1} \F{=1} (\bigwedge_i {\xi_i}) \\
        &\overset{\eqref{eq:C}}{\equiv}
            \G{=1} \F{=1} (\bigwedge_{l \in C}{l} \land \bigwedge_i
            {\G{=1} \vartheta_i} \land \bigwedge_i {\F{\rhd r}
            \zeta_i}) \\
        &\overset{\eqref{eq:GFCf}}{\equiv}
            \G{=1} (\F{=1} (\bigwedge_{l \in C}{l} \land \bigwedge_i
            {\G{=1} \vartheta_i}) \land \bigwedge_i {\F{=1} \zeta_i})
            \\
        &\overset{\eqref{eq:GC}}{\equiv}
        \G{=1} (\F{=1}(\bigwedge_{l \in C}{l} \land \G{=1} \bigwedge_i
            {\vartheta_i}) \land \bigwedge_i {\F{=1} \zeta_i}) \\
        &\overset{\eqref{eq:GFCG}}{\equiv}
            \G{=1} (\F{=1} (\bigwedge_{l \in C}{l}) \land \F{=1}
            \G{=1} (\bigwedge_i {\vartheta_i}) \land \bigwedge_i
            {\F{=1} \zeta_i}) \\ 
        &\overset{o.I.H}{\equiv_{fin}}
            \G{=1} (\F{=1} (\bigwedge_{l \in C}{l}) \land \bigwedge_i
            {\F{=1} \zeta_i} \\
        &\land
            \F{=1} \G{=1} (\bigwedge_{l \in A_{\vartheta}}{l} \land
            \F{=1} \G{=1} (\bigwedge_{l \in B_{\vartheta}}{l}) \land
            \bigwedge_{i \in I_{\vartheta}}{\F{=1} (\bigwedge_{l \in
            C_{\vartheta, i}}{l})})) \\ 
        &\overset{\eqref{eq:FGC}}{\equiv}
            \G{=1} (\F{=1} (\bigwedge_{l \in C}{l}) \land \bigwedge_i
            {\F{=1} \zeta_i} \land \F{=1} \G{=1} (\bigwedge_{l \in
            A_{\vartheta}}{l}) \\
        &\land
            \F{=1} \G{=1} \F{=1} \G{=1} (\bigwedge_{l \in
            B_{\vartheta}}{l}) \land \bigwedge_{i \in
            I_{\vartheta}}{\F{=1} \G{=1} \F{=1} (\bigwedge_{l \in
            C_{\vartheta, i}}{l})}) \\ 
        & \overset{\substack{\eqref{eq:FGF} \\
        \eqref{eq:GFG}}}{\equiv}
            \G{=1} (\F{=1} (\bigwedge_{l \in C}{l}) \land \bigwedge_i
            {\F{=1} \zeta_i} \land \F{=1} \G{=1} (\bigwedge_{l \in
            A_{\vartheta}}{l}) \\
        &\land
            \F{=1} \G{=1} (\bigwedge_{l \in B_{\vartheta}}{l}) \land
            \bigwedge_{i \in I_{\vartheta}}{\G{=1} \F{=1}
            (\bigwedge_{l \in C_{\vartheta, i}}{l})}) \\ 
        &\overset{\eqref{eq:FGC}}{\equiv}
            \G{=1} (\F{=1} (\bigwedge_{l \in C}{l}) \land \bigwedge_i
            {\F{=1} \zeta_i} \land \F{=1} \G{=1} (\bigwedge_{l \in
            A_{\vartheta} \cup B_{\vartheta}}{l}) \\
        &\land
            \bigwedge_{i \in I_{\vartheta}}{\G{=1} \F{=1}
            (\bigwedge_{l \in C_{\vartheta, i}}{l})}) \\ 
        &\overset{\substack{\eqref{eq:GC} \\ \eqref{eq:GG}}}{\equiv}
            \G{=1} (\bigwedge_i {\F{=1} \zeta_i} \land \F{=1} \G{=1}
            (\bigwedge_{l \in A_{\vartheta} \cup B_{\vartheta}}{l})
            \land \F{=1} (\bigwedge_{l \in C}{l}) \land \bigwedge_{i
            \in I_{\vartheta}}{\F{=1} (\bigwedge_{l \in C_{\vartheta,
            i}}{l})}) \\ 
        &\overset{\substack{\eqref{eq:GC} \\ i.I.H}}{\equiv}
            \G{=1} (\bigwedge_i {(\F{=1} \G{=1} (\bigwedge_{l \in
            B_{\zeta_i}}{l}) \land \bigwedge_{j \in
            I_{\zeta_i}}{\F{=1} (\bigwedge_{l \in C_{\zeta_i,
            j}}{l})})} \\
        &\land \F{=1} \G{=1}
            (\bigwedge_{l \in A_{\vartheta} \cup B_{\vartheta}}{l})
            \land \F{=1} (\bigwedge_{l \in C}{l}) \land \bigwedge_{i
            \in I_{\vartheta}}{\F{=1} (\bigwedge_{l \in C_{\vartheta,
            i}}{l})}) \\ 
        &\overset{\eqref{eq:FGC}}{\equiv}
            \G{=1} (\F{=1} \G{=1} ((\bigwedge_i {\bigwedge_{l \in
            B_{\zeta_i}}{l}}) \land (\bigwedge_{l \in A_{\vartheta}
            \cup B_{\vartheta}}{l})) \land  \\
        &\land
            \bigwedge_i {\bigwedge_{j \in I_{\zeta_i}}{\F{=1}
            (\bigwedge_{l \in C_{\zeta_i, j}}{l})}} \land \F{=1}
            (\bigwedge_{l \in C}{l}) \land \bigwedge_{i \in
            I_{\vartheta}}{\F{=1} (\bigwedge_{l \in C_{\vartheta,
            i}}{l})})
    \end{align*}
    \endgroup
    \QED
\end{proof}

\end{document}